\documentclass{article}
\usepackage{kr}
\usepackage{thmtools}
\usepackage{thm-restate}

\usepackage{times}
\usepackage{soul}
\usepackage{url}
\usepackage[hidelinks]{hyperref}
\usepackage[utf8]{inputenc}
\usepackage[small]{caption}
\usepackage{graphicx}
\usepackage{amsmath}
\usepackage{amsthm}
\usepackage{amssymb}
\usepackage{booktabs}
\usepackage{algorithm}
\usepackage{algorithmic}
\usepackage{adjustbox}
\usepackage{makecell}
\usepackage{listings}
\usepackage{tikz}
\usepackage{array}
\usetikzlibrary{arrows.meta}
\usetikzlibrary{automata,positioning}
\usetikzlibrary{calc}
\newtheorem{example}{Example}

\newcommand{\anonymize}[2]{#1}
\newcommand{\anonAuth}[1]{\anonymize{#1}{Anonymous Author}}
\newcommand{\anonAff}[1]{\anonymize{#1}{Anonymous Affiliation}}
\newcommand{\anonEmails}[1]{\anonymize{#1}{Anonymous Emails}}

\title{NILE:  Formalizing Natural-Language Descriptions of Formal Languages}

\author{%
    \anonAuth{Tristan Kneisel}$ ^1 $\and
    \anonAuth{Marko Schmellenkamp}$ ^1 $\and
    \anonAuth{Fabian Vehlken}$^1$\and
    \anonAuth{Thomas Zeume}$^1$\\
\affiliations
$^1$\anonAff{Ruhr University Bochum, Germany}\\
\emails
\anonEmails{\{tristan.kneisel, marko.schmellenkamp, fabian.vehlken, thomas.zeume\}@rub.de}}

\newif\ifcomments
\newif\ifchanges
\commentsfalse\changesfalse
\commentstrue\changestrue

\usepackage{xspace}
\usepackage[capitalise]{cleveref}

\newcommand{\N}{\ensuremath{\mathbb{N}}}

\newcommand{\df}{\ensuremath{\mathrel{\smash{\stackrel{\scriptscriptstyle{
    \text{def}}}{=}}}} \;}

\newcommand  {\myclass} [1]  {\ensuremath{\textsf{\upshape #1}}}

\newcommand{\StaClass}[1]{\myclass{#1}\xspace}

\newcommand{\CQ}[1][]{\StaClass{CQ}}
\newcommand{\UCQ}[1][]{\StaClass{UCQ}}
\newcommand{\CQneg}[1][]{\StaClass{CQ\ensuremath{^{\mneg}}}}
\newcommand{\UCQneg}[1][]{\StaClass{UCQ\ensuremath{^{\mneg}}}}

\theoremstyle{plain}

\theoremstyle{definition}

\newtheorem*{question*}{Question}
\newtheorem*{openquestion*}{Open question}

\newenvironment{proofsketch}{\begin{proof}[Proof sketch.]}{\end{proof}}

\providecommand {\calD}      {{\mathcal D}\xspace}

\providecommand {\calI}      {{\mathcal I}\xspace}
\providecommand {\calL}      {{\mathcal L}\xspace}
\providecommand {\calM}      {{\mathcal M}\xspace}

\providecommand {\calR}      {{\mathcal R}\xspace}
\providecommand {\calS}      {{\mathcal S}\xspace}

\ifcomments
\newcommand{\commentbox}[1]{\noindent\framebox{\parbox{0.98\linewidth}{#1}}}

\newcommand{\acomment}[2]{\ \\ \fbox{\parbox{0.98\linewidth}{{\sc #1}: #2}}}
\newcommand{\mcomment}[2]{{\color{blue}(#1)}\footnote{#1: #2}} %
\else
\newcommand{\commentbox}[1]{}
\newcommand{\mcomment}[2]{}
\newcommand{\acomment}[2]{}
\fi

\ifchanges

\setul{}{0.2mm}
\setstcolor{red}

\else

\fi

\newif\ifwithappendix

\withappendixtrue

\ifwithappendix

\newcommand{\insertIfWithAppendix}[1]{#1}
\else
\newcommand{\insertIfWithAppendix}[1]{\!}
\fi

\newcommand{\Nile}{\textsc{Nile}\xspace}
\newcommand{\Mosel}{\textsc{Mosel}}
\newcommand{\iffdef}[0]{\overset{\text{def}}{\iff}}
\newcommand{\eqdef}[0]{\overset{\text{def}}{=}}
\newcommand{\eqvdef}[0]{\overset{\text{def}}{\equiv}}

\begin{document}

\maketitle

\begin{abstract}
  This paper explores how natural-language descriptions of formal languages can be compared to their formal representations and how semantic differences can be explained. This is motivated from educational scenarios where learners describe a formal language (presented, e.g., by a finite state automaton, regular expression, pushdown automaton, context-free grammar or in set notation) in natural language, and an educational support system has to (1) judge whether the natural-language description accurately describes the formal language, and to (2) provide explanations why descriptions are not accurate.

  To address this question, we introduce a representation language for formal languages, \textsc{Nile}, which is designed so that \textsc{Nile} expressions can mirror the syntactic structure of natural-language descriptions of formal languages. \textsc{Nile} is sufficiently expressive to cover a broad variety of formal languages, including all regular languages and fragments of context-free languages typically used in educational contexts. Generating \textsc{Nile} expressions that are syntactically close to natural-language descriptions then allows to provide explanations for inaccuracies in the descriptions algorithmically.

  In experiments on an educational data set, we show that LLMs can translate natural-language descriptions into equivalent, syntactically close \textsc{Nile} expressions with high accuracy -- allowing to algorithmically provide explanations for incorrect natural-language descriptions. Our experiments also show that while natural-language descriptions can also be translated into regular expressions (but not context-free grammars), the expressions are often not syntactically close and thus not suitable for providing explanations.

\end{abstract}

\section{Introduction}\label{section:introduction}
Learning to construct formal representations for scenarios described in natural language and to understand such formal  representations is an important part of CS higher education \cite{ACM2024,GI2016}. For instance, typical courses on formal foundations of CS at the Bachelor level cover formal language representations such as automata (deterministic and non-deterministic finite automata, push-down automata, and Turing machines), regular expressions, and formal grammars.  Two exercise types likely to be encountered by students are the following:
\begin{description}
	\item[(E1)] Construction: Given a natural description of a language $L$, construct a formal representation (automaton, regular expression, grammar, \dots) $\mathcal{R}$ such that $\mathcal{L}(\mathcal{R}) = L$.
	\item[(E2)] Interpretation: Given a formal representation $\mathcal{R}$, provide a natural description $\mathcal{D}$ of its language $\mathcal{L}(\mathcal{R})$.
\end{description}
While construction exercises are broadly covered for formal languages by educational support systems such as  AutomataTutor \cite{AlurDGKV2013,DAntoniHKRW20}, Iltis \cite{SchmellenkampVZ24}, and JFLAP \cite{Rodger06}, there has been no work on helping students at interpreting formal representations in such systems to the best of our knowledge. This is likely because processing natural-language submissions is more challenging than processing submissions of formal representations.

This paper aims to address this gap by laying the groundwork for exercises supporting students to describe formal languages in educational support systems.

For supporting such educational tasks, natural-language descriptions $\mathcal{D}$ (provided by students) for a formal representation $\mathcal{R}$ (provided as part of the assignments by instructors) need to be checked for correctness and, if necessary, suitable explanations for inaccurate descriptions be provided. A na\"{\i}ve approach is to (a) translate $\mathcal{D}$ into a formal representation $\mathcal{R}_\mathcal{D}$, (b) check for equivalence of  $\mathcal{R}$ and  $\mathcal{R}_\mathcal{D}$, and (c) compute feedback using $\mathcal{D}$, $\mathcal{R}$ and $\mathcal{R}_\mathcal{D}$. Experiments (see Section \ref{section:evaluation}) show that while current LLMs perform well on Task (a) for regular expressions, they do not perform so well for context-free representations. Furthermore,  regular expressions generated for natural-language descriptions also tend to be syntactically far away from the descriptions, leaving it unclear how to explain inaccuracies to students. %

\paragraph*{Contributions}

We address this shortcoming of existing formal representations of formal languages by introducing a descriptive formal representation language, \textsc{Nile}, whose expressions can both (i) mirror closely  the syntactical structure of natural-language descriptions and (ii) be easily translated into conventional formal representations.
Due to these properties, \Nile is well suited to check student solution attempts and provide explanations.
Natural language descriptions $\mathcal{D}$ (provided by students) can be translated into \textsc{Nile} expressions $\mathcal{N}_\mathcal{D}$ whose syntactic structure is close to $\mathcal{D}$. %
Correctness of solution attempts can be checked by testing equivalence of $\mathcal{N}_\mathcal{D}$ and a \textsc{Nile} expression $\mathcal{N}$ for the solution language (provided by an instructor). 
If necessary, explanations can be provided by analysing differences of the syntax trees of $\mathcal{N}_\mathcal{D}$ and $\mathcal{N}$ with suitable algorithms.

We implement \Nile and show that it captures all regular languages and bounded context-free languages. The latter class includes many typical context-free languages used in educational settings, such as $\{a^nb^m(ab)^{n+m}\mid m,n\in\N\}$.

We then explore how well LLMs can test correctness of natural-language descriptions of formal languages and translate them into regular expressions, context-free grammars, and \Nile expressions. Our focus is on translations that preserve the syntactic structure (as such translations are the basis for explaining semantic differences) and whether \Nile performs better than classical representations.

\newtheorem{researchquestion}{Research question}

\begin{researchquestion}
	How well can LLMs decide whether a natural-language description fits to a formal language?
\end{researchquestion}

\begin{researchquestion}
	How well can LLMs construct a formal representation that semantically matches a natural-language description?
\end{researchquestion}
\begin{researchquestion}[Main]
	How well can LLMs construct a formal representation that syntactically matches a natural-language description?
\end{researchquestion}

We address the questions by evaluating answers from OpenAI’s gpt-oss-120b model on 2,280 natural-language descriptions by students for 6 regular and 5 context-free languages in different formal representations (see \cref{table:exercises}).

The results are summarized in \cref{table:summary}. In a nutshell, the LLM performs well in deciding how well a natural-language description fits a formal language and it can also translate natural-language descriptions into the formal representations, with translations to \Nile performing slightly better than translations into regular expressions and context-free languages. However, the syntactic structure of natural-language descriptions is only matched with high accuracy if they are translated into \Nile expressions.

\paragraph*{Related Work} \textsc{Mosel} is another language for describing regular languages with highlevel constructs in educational settings \cite{AlurDGKV2013}. Designed for providing feedback and grading automata constructions, its expressions tend to be syntactically different from natural-language descriptions as it is based on monadic second-order logic extended by indexing operators.  Small Bert-based models have been used to provide feedback for educational tasks on designing logical vocabularies \cite{KneiselVZ25}. %

There has also been work on translating natural language into regular expressions, see e.g.\ \cite{ZhangGCS23}. %
The focus there is typically to find correct regular expressions, mostly for usage in code. In contrast to our use case, staying close to the natural-language description is not a priority.

\section{The \Nile language}\label{section:nile}
The design goals for \Nile (short for ``Natural, Intuitive Language Expressions'') are (i) compositionality, i.e.~simple expressions easily compose to more complex expressions, (ii) expressive power, i.e.~a broad variety of typical regular and context-free languages from educational contexts can be expressed, (iii) closeness to natural language, i.e.~constructs used in natural-language descriptions can be mirrored in the syntax of \Nile, (iv) extendability, i.e.~further syntactical constructs to cover natural language constructs can be easily added, and (v) algorithmic accessibility, i.e.~the language allows for easy conversion into other formal representations and standard algorithmic problems are (efficiently) solvable.

\Nile uses language expressions (atoms and composite expressions) to describe formal languages and unary number predicates to describe sets of natural numbers (which can be used in language expressions). Formally, language expressions are evaluated wrt.~a string $w$ and an interpretation $\beta$ of free natural number variables, and unary number predicates are evaluated wrt.~a natural number $n$ and $\beta$. For simplicity, we denote $\bar w \df (w, \beta)$ and $\bar n \df (n, \beta)$ for $w \in \Sigma^*$, $n \in \N_0$ and some interpretation $\beta$.

The syntax and semantics of \Nile language expressions and unary number predicates is as follows:
\begin{itemize}
	\item[(1)] \Nile atoms: $\bar w \models \sigma_1\sigma_2...\sigma_k \iffdef w = \sigma_1\sigma_2...\sigma_k \in \Sigma^*$

	\item[(2)] Composite \Nile expressions using $\Nile$ expressions $\varphi$ and $\psi$, and a unary number predicate $P$:
		\begin{itemize}
			\item[(a)] Negations: $\bar w \models \lnot \varphi \iffdef \bar w \not\models \varphi$
			\item[(b)] Disjunctions: $\bar w \models \varphi \lor \chi \iffdef (\bar w \models \varphi) \lor (\bar w \models \chi)$
			\item[(c)] Concatenations: $\bar w \models \varphi \circ \chi \iffdef \exists i \in [1, \left|w\right| + 1]:$ \\
					\hspace*{0.2cm} $((w[1, i - 1], \beta) \models \varphi) \land ((w[i, \left|w\right|], \beta) \models \chi)$
			\item[(d)] Repetitions: $\bar w \models \mathtt{REP}(P, \varphi)$\\ $\hspace*{1.6cm}\iffdef \exists n \in \mathbb{N}_0:  (\bar n \models P) \land (\bar w \models \bigcirc_{i=1}^n \varphi)$
			\item[(e)] Counted occurences: $\bar w \models \mathtt{HAS}(P, \varphi) \iffdef \left|I\right| \models P$
 	\\ \hspace*{0.2cm}\text{where } $I \eqdef \{i \in [1, \left|w\right| + 1] \mid (w[i, \left|w\right|], \beta) \models \varphi \circ \top\}$

			\item[(f)] Quantification of numbers:\\
		$\bar w \models \exists x~\varphi \iffdef \exists n \in \N_0: (w, \beta[x/n]) \models \varphi$

		\end{itemize}
		\item[(3)] Unary number predicates are Boolean combinations of
				\begin{itemize}
				\item $\bar n \models (\ge t) \iffdef n \ge t[\beta]$
				\item $\bar n \models (\equiv t \mod c) \iffdef n \equiv t[\beta] \mod c $
				\end{itemize}
		where $t$ is a term of the form $\sum_i c_ix_i + d$ with number constants $c_i$, $d$ and variables $x_i$;  $c$ is a number constant; and $t[\beta]$ is the evaluation of $t$ under interpretation $\beta$ of the free variables.
\end{itemize}

We use the usual syntactic sugar $\bot$, $\top$, $\wedge$, $\rightarrow$, $\leftrightarrow$ resp. $>$, $=$, $<$, $\le$ for Boolean combinations of \Nile language expressions and number predicates.

The above constructs constitute the \emph{core \Nile language}. The \emph{extended \Nile language} allows further syntactic sugar definable within the core language, for instance:

\begin{itemize}
 \item $\mathtt{BEGIN}(P,\varphi) \eqvdef \mathtt{REP}(P,\varphi) \circ ( \varepsilon \lor \lnot(\varphi \circ \top))$
	\item $\mathtt{COUNT}(\varphi) \hspace{-1mm}>\hspace{-1mm} \mathtt{COUNT}(\psi)\eqvdef \exists x \big(\mathtt{HAS}(>\hspace{-1mm}x, \varphi)\wedge \mathtt{HAS}(\leq\hspace{-1mm}x, \psi)\big)$
\end{itemize}

\newcommand{\red}[1]{\textcolor{red}{#1}}
\newcommand{\black}[1]{\textcolor{black}{#1}}
We give some examples of authentic \Nile expressions within the context of our educational application.%

\begin{example}\label{example:nile:has}
The regular language given by the automaton
		\scalebox{0.5}{
		\tikz[baseline=-0.5ex,shorten >=1pt, node distance=1.0cm, on grid, auto,
		state/.style={circle, draw, minimum size=4.5mm, inner sep=0pt}
		]{		\node[state] (q0) {$q_0$};
				\draw[->] ($(q0.south)+(0mm,-3mm)$) -- (q0.south);
				\node[state] (q1) [right=of q0] {$q_1$};
				\node[state] (q2) [right=of q1] {$q_2$};
				\node[state, accepting] (q3) [right=of q2] {$q_3$};

				\path[->]
					(q0) edge[loop left] node {$a,b$} ()
								edge[] node {$a$} (q1)
					(q1) edge node {$a$} (q2)
					(q2) edge node {$a$} (q3)
					(q3) edge[loop right] node {$a,b$} ();
		}} may  be (wrongly) described by a student as  ``Strings with three $a$\textup{'s}'' expressible in \Nile as
		\[\mathtt{HAS}(=3,a)\]
		Comparison with a correct description ``Strings with three consecutive $a$\textup{'s}'' expressible in \Nile as
			\[\mathtt{HAS}(\red{\mathtt{REP}(\black{=3, a})})\]
		yields the explanation that the consecutiveness has not been taken into account. This explanation can be derived algorithmically from the syntax trees of the \Nile expressions.
\end{example}
\begin{example}\label{example:nile:quantifier}
The context-free language given by the grammar $S \rightarrow aSab \mid X; X \rightarrow bXab \mid \varepsilon$ may be (wrongly) described as  ``Strings starting with an $a$, followed by a sequence of $b$\textup{'s}, followed by a sequence of $ab$\textup{'s}, such that there is one more $ab$ than $b$\textup{'s}'' expressible in \Nile as
			\[\exists i[a \circ \mathtt{REP}(=i,b) \circ \mathtt{REP}(=i+1,ab)]\]
		Comparison with a correct description ``Strings starting with $a$\textup{'s}, followed by a sequence of $b$\textup{'s}, followed by a sequence of $ab$\textup{'s}, such that there are exactly as many $ab$\textup{'s} as the sum of the number of $a$\textup{'s} and $b$\textup{'s}'' expressible in $\Nile$ as
		\[\exists i\red{\exists j} [\red{\mathtt{REP}(=j,\black{a})} \circ \mathtt{REP}(=i,b) \circ \mathtt{REP}(=i+\red{j},ab)]\]
		yields an explanation based on the counts of $a$\textup{'s}, $b$\textup{'s}, and $ab$\textup{'s} derivable from the \Nile expressions.
\end{example}

We next explore the expressive power of \Nile wrt. classical formal language representations. \emph{Regular \textsc{Nile}} contains all rules of core-\textsc{Nile} except for the quantification rule $\exists x \varphi$.

\begin{restatable}{theorem}{theoremNileRegular}
\label{theorem:nile:regular}
	A language is expressible in regular \Nile iff it is regular.%

\end{restatable}

Core-\Nile can express all bounded context-free languages, in particular the non-regular language $\{a^nb^n \mid n \in \N\}$.  Here, a language $L \subseteq \Sigma^*$ is called \emph{bounded} if there are $w_1, \ldots, w_m \in \Sigma^*$ so that $ L \subseteq \calL(w_1^* \ldots w_m^*)$.

\begin{restatable}{theorem}{theoremNileBounded}
\label{theorem:nile:bounded}
	Every bounded context-free language is expressible in core-\Nile.
\end{restatable}

Observe that there are also languages expressible in core-\Nile that are not context-free, e.g. $\{a^nb^nc^n \mid n \in \N_0\}$.

\section{Evaluation}\label{section:evaluation}

In this section we address our research questions by exploring how well LLMs can translate natural-language descriptions into classical representations of formal languages as well as into \Nile expressions. We describe our data set and its collection in Section \ref{section:dataset} and evaluate the dataset with respect to our three research questions in Section \ref{section:results}.

\subsection{Data set and Annotation}\label{section:dataset}
We describe (i) our educational data set, (ii) how LLM labels regarding our research questions were gathered, and (iii) how the resulting data was annotated for the evaluation.

\begin{table}
  \setlength{\tabcolsep}{2pt}
  \def\arraystretch{1.1}
	\small\centering
	\scalebox{0.85}{%
	\begin{tabular}{l>{}p{6.8cm}lrr}
	       & \small Language & \makebox[6mm][l]{Format} & $\#$ & \makebox[7mm][l]{$\# \calR_\calD$}\\
		\hline
		R1 & $\{ab^n \mid n \ge 1\}$ %
		  & Set & 276 & 28\\
		R2 & $\{\varepsilon\} \cup \{wa \mid w \in \{a, b\}^*\}$ %
		  & RE & 273 & 22\\
		R3 & $\{w \in \{a, b\}^* \mid w \text{ contains } aaa \}$ %
		  & NFA & 255 & 29\\
		C1 & $\{w \in \{a, b\}^* \mid \#_b(w) \ge \#_a(w)\}$ %
		  & CFG & 209 & 45\\
		C2 & $\{a^nb^n \mid n \in \mathbb{N}_0 \}$ %
		  & PDA & 201 & 17\\
		  \hline
		R4 & $\{a^mb^n \mid m \text{ even}, n \text{ odd} \}$ %
		  & Set & 205 & 18\\
		R5 & \mbox{$\{w \in \{a, b, c\}^* \mid \text{every } a \text{ is directly followed by a } b\}$}
			& DFA & 198& 38 \\
		R6 & $\{w \in \{a, b, c\}^* \mid \#_b(w) = 0 \lor w \text{ ends in } aba\}$ %
		  & NFA & 200 & 16\\
		C3 & $\{a^mb^n(ab)^{m+n} \mid m, n \in \mathbb{N}_0 \}$ %
		   & CFG & 162 & 35\\
		C4 & $\{wba^n \mid w \in \{a, b\}^*, n \in \mathbb{N}_0, \#_a(w) = n\}$ %
		   & CFG & 159 & 50\\
		C5 & $\{w \in \{a, b, c\}^* \mid \#_a(w) \le 2 \cdot \#_b(w)\}$ %
		   & PDA & 142 & 41\\
	   \hline
		$\Sigma$ & & & 2280 \\
    \end{tabular}%
	}
	\caption{Exercises in the dataset including format in which the language is presented to students, number of submissions ($ \# $), and number of non-equivalent \Nile expressions produced by the LLM in our experiments ($ \# \calR_\calD $) for the submissions.}
	\label{table:exercises}
\end{table}

\paragraph{Educational data set} Our data set contains 2,280 natural-language descriptions by students for 11 formal languages in different representations (set notations, regular expressions, finite state and pushdown automata, context-free grammars) provided by instructors, see \cref{table:exercises}. All languages are expressible in \Nile. Data for languages (R1) -- (R3), (C1), and (C2) was collected in winter term 2024 and used for designing experiments and prompts. Data for the remaining six exercises was collected in winter term 2025 and not examined before the evaluation.

\begin{table*}[th]
    \centering\small
    \renewcommand{\arraystretch}{1.0}
    \setlength{\tabcolsep}{5.3pt}
	\scalebox{1}{%
\begin{tabular}{l||>{\bfseries}r>{\bfseries}r|rrrr|rrrr}
	\toprule
	& \multicolumn{2}{c|}{\textbf{all}} & \multicolumn{4}{c|}{Regular languages} & \multicolumn{4}{c}{Context-free languages}  \\
	& \multicolumn{2}{c|}{} & \multicolumn{2}{c}{Design data} & \multicolumn{2}{c|}{Validation data} & \multicolumn{2}{c}{Design data} & \multicolumn{2}{c}{Validation data} \\
	& $\#$ & $\%$ & $\#$ & $\%$ & $\#$ & $\%$ & $\#$ & $\%$ & $\#$ & $\%$ \\
	\midrule
	\hspace{1mm} all descriptions & 2040 & 100\% & 711 & 100\% & 585 & 100\% & 348 & 100\% & 396 & 100\% \\
	\midrule
	\multicolumn{11}{l}{\textbf{RQ 1}: How well can LLMs decide whether a natural-language description fits to a formal language?} \\
	\midrule
	\hspace{1mm} (M1) Directly & 1956 & 95.88\% & 676 & 95.08\% & 564 & 96.41\% & 343 & 98.56\% & 373 & 94.19\%\\
	\hspace{1mm} (M2) Via classical representation (RE / CFG) & 1913 & 93.77\% & 704 & 99.02\% & 572 & 97.78\% & 309 & 88.79\% & 328 & 82.83\%\\
	\hspace{1mm} (M3) Via \Nile & 1950 & 95.59\% & 699 & 98.31\% & 574 & 98.12\% & 347 & 99.71\% & 330 & 83.33\%$^*$\hspace*{-1.5mm}\\
	\midrule
	\multicolumn{11}{l}{\textbf{RQ 2}: How well can LLMs construct a formal representation that semantically matches a natural-language description?} \\
	\midrule
	\hspace{1mm} (M2) Classical representation (RE / CFG) & 1632 & 80.00\% & 681 & 95.78\% & 484 & 82.74\% & 253 & 72.70\% & 214 & 54.04\%\\
	\hspace{1mm} (M3) \Nile expression & 1881 & 92.21\% & 693 & 97.47\% & 566 & 96.75\% & 341 & 97.99\% & 281 & 70.96\%$^{**}$\hspace*{-3mm}\\
	\midrule
	\multicolumn{11}{l}{\textbf{RQ 3}: How well can LLMs construct a formal representation that syntactically matches a natural-language description?} \\
	\midrule
	\hspace{1mm} (M2) Classical representation (RE / CFG) & 493 & 24.17\% & 385 & 54.15\% & 104 & 17.78\% & 3 & 0.86\% & 1 & 0.25\%\\
	\hspace{1mm} (M3) \Nile expression & 1742 & 85.39\% & 679 & 95.50\% & 474 & 81.03\%\!$^{\triangle}$\hspace*{-2.3mm} & 334 & 95.98\% & 255 & 64.39\%$^{***}$\hspace*{-4.5mm}\\
	\bottomrule
\end{tabular}
	}
\caption{Summary of the evaluation of RQs 1 to 3. Overall performance, performance for regular languages R1 -- R6 and performance for context-free languages C1 -- C5 is shown. For RQ 1, the LLM was asked to judge correctness directly or to construct a formal representation (RE, CFG, \Nile) which was then checked for correctness algorithmically. For RQs 2 and 3, the LLM was asked to construct a semantically and syntactically matching formal representation. Only design data (languages R1 -- R3, C1, C2) has been used for conceptual and prompt design. Posthoc adjusting the prompt for \Nile with two additional examples for non-covered idioms increases the performance to 87.86\%\!$^\triangle$, 93.69\%$^*$, 85.86\%$^{**}$, and 80.30\%$^{***}$. A more detailed evaluation can be found \insertIfWithAppendix{in \cref{section:appendix-detailed-evaluation-data}} in the appendix.}
\label{table:summary}
\end{table*}

\paragraph*{LLM labels} For each data point $(\text{id}, \calR, \calD, \calS)$ with  formal representation $\calR$, natural-language description $\calD$ by an instructor, and description $\calS$ by a student, an LLM was prompted for the following information:
\begin{itemize}
	\item \textbf{Correctness (M1):} Do the natural-language descriptions $\calD$ and $\calS$ describe the same formal language?
	\item \textbf{Classical representation (M2):} Given the natural-language description $\calS$, what is an equivalent regular expression or context-free grammar (depending on $\calR$) that syntactically matches the description $\calS$?
	\item \textbf{\Nile representation (M3):} Given the natural-language description $\calS$, what is an equivalent \Nile expression that syntactically matches the description $\calS$?

\end{itemize}
For invalid or inexpressible natural-language descriptions, the model was asked to return an error label. The alphabet of the language $\calL(\calR)$ was included in the prompt for all (M1) --  (M3), as it is often necessary to interpret natural-language descriptions. For (M2), regular expressions and context-free grammars were chosen because they are descriptive representations. Each prompt included only few examples ($\leq 10$), as these representations are in the training data of models. For (M3), the prompt described the language \Nile via 62 examples covering all \Nile~constructs and important idioms. No formal specification was given. %

As LLM we used the gpt-oss-120b model \cite{Openai2025} with temperature 0. It was chosen after performing initial experiments with other locally deployable models (Mistral-Small-3.2-24B-Instruct-2506, openai/gpt-oss-20b). %

\paragraph*{Annotation}

 For each data point $(\text{id}, \calR, \calD, \calS)$ with answers (M1) -- (M3) by the model, we annotated the model's answers. To this end, let $\calI(\calS)$ be the set of formal languages that are plausible interpretations for the student description~$\calS$. As preliminary step, all data points with $\calI(\calS) = \emptyset$ and non-natural-language descriptions such as (operational descriptions of) formal representations have been removed (240 of 2,280).  The natural-language descriptions $\calS$ by students for the remaining 2,040~data points are annotated for correctness, i.e. $\mathcal{L}(\mathcal{R}) \in \mathcal{I}(\mathcal{S})$.

The LLM answers (M2)  and (M3) are parsed as regular expressions, context-free grammars, or \Nile expressions, respectively. If successful, the semantic correctness (i.e. is $\mathcal{L}(\mathcal{R}_\mathcal{M}) \in \mathcal{I}(\mathcal{S})$?) and syntactic correctness (i.e. if semantically correct, does $\calS$ syntactically match~$\mathcal{R}_\mathcal{M}$?) of the resulting representation $\calR_\calM$ wrt. $\calS$ are annotated.

\subsection{Results}\label{section:results}
The evaluation results are summarized in Table \ref{table:summary}.

For answering RQ 1, we evaluate how often the LLM decides correctly whether a natural-language description is accurate. The accuracy is $\approx$\,95\% when only asking for a decision (M1), and similar for (M2) and (M3), where the LLM is asked to construct a formal representation (RE, CFG, \Nile) for which equivalence is checked algorithmically.

For answering RQ 2, we evaluate how often the LLM produces a correct formal representation in $\calI(\calS)$ when asked for a syntactically matching, equivalent representation for a  description $\calS$ (M2 and M3). Over the whole data set, natural-language descriptions are translated with accuracy 92.21\% into \Nile and with accuracy 80.00\% into REs/CFGs. Performance for CFGs is significantly lower.

For answering RQ 3, we evaluate how often the LLM produces a correct formal representation of a language in $\calI(\calS)$ that also syntactically matches the natural-language description $\calS$. Translations into \Nile expressions, with accuracy $>$85\%, outperform translations into REs/CFGs by far. The very low accuracy for CFGs is not surprising, as CFGs represent many simple natural-language constructs, such as counting, by recursion. Our data shows that \Nile addresses this shortcoming and achieves its main design goal.  %

\paragraph{Limitations} The accuracy of the LLM for all representations and all RQs strongly depends on the concrete formal languages (see \insertIfWithAppendix{\cref{section:appendix-detailed-evaluation-data} in the} appendix). This is likely a contributing factor to differing accuracy across design and validation data set. Inspection of the validation data set after the experiments also revealed that the prompt of the LLM did not cover all \Nile constructs; adding two examples to the prompt significantly increases the accuracy of \Nile for all research questions (see \cref{table:summary}).

\section{Summary and Discussion}\label{section:summary}
We introduced \Nile %
and showed that natural-language descriptions of formal languages can be translated into syntactically close \Nile expressions with LLMs.%

\Nile can then be used for providing feedback on natural-language descriptions of formal languages by students. For example, recall the \Nile expression $\mathtt{HAS}(=3,a)$ for ``Strings with three $a$s'' from Example \ref{example:nile:has} wrongly expressing the language ``Strings with three consecutive $a$s''. The closeness of $\mathtt{HAS}(=3,a)$ to an accurate \Nile expression $\mathtt{HAS}(\red{\mathtt{REP}(\black{=3, a})})$ can be exploited for explanations, for instance by providing a derived natural-language description focussing on the mistake, e.g., \textit{Do you really mean ``Strings that contain three $a$\textup{'s} at arbitrary positions''?}.

While \Nile captures all regular and bounded context-free languages, it does not capture context-free languages such as the palindrome language $\{w \mid w = w^R\}$. However, the modularity of \Nile allows to easily add further language constructs, for instance, word quantification such as $\exists w [a \circ w \circ w]$ for the language $\{aww \mid w \in \Sigma^*\}$ or word reversal $\mathtt{REVERSE}(\varphi)$ for all words $w$  such that $w^R \models \varphi$. Our implementation already includes these constructs.

The next step is to design didactic explanations and incorporate them into an educational support system. We also plan to use the approach for more formalisms, e.g.~natural-language descriptions of the meaning of logical formulas.
 
\newpage

\bibliographystyle{kr}
\bibliography{bibliography}

\ifwithappendix
\newpage
\appendix
\section{Appendix}
\newcommand{\moselat}[0]{\mathbin{@}}
\newcommand{\moselperc}[0]{\mathbin{\%}}
\newcommand{\lqm}{\mathord{\mbox{`}}}
\newcommand{\rqm}{\mathord{\mbox{'}}}

\subsection{Additional Material for Section \ref{section:nile}}
\subsubsection{Additional Syntactic Sugar}

\begin{itemize}
	\item Logical
		\begin{itemize}
		 \item $\forall i \varphi \eqvdef \neg \exists i \neg \varphi$
		\end{itemize}

	\item Elementary
		\begin{itemize}
			\item $\mathtt{REP}(\varphi) \eqvdef \mathtt{REP}(\geq 0, \varphi)$
			\item $\mathtt{HAS}(\varphi) \eqvdef \mathtt{REP}(\geq 1, \varphi)$
			\item $\mathtt{LEN}=P \eqvdef \mathtt{REP}(P, \bigvee_{\sigma \in \Sigma} \sigma)$
			\item $\mathtt{ALPH}(\sigma_1, \ldots, \sigma_k){:}\ \varphi \eqvdef \mathtt{REP}(\sigma_1 \lor \ldots \lor \sigma_k) \land \varphi$
			\item $\mathtt{ALTERNATE}(\varphi, \chi) \eqvdef (\varphi \lor \varepsilon) \circ \mathtt{REP}(\chi \circ \varphi) \circ (\chi \circ \varepsilon)$
			\item $\mathtt{CONS}(P_1, \varphi_1;\ldots; P_k, \varphi_k)$\\ $\eqvdef \mathtt{REP}(\varphi_1 \lor \ldots \lor \varphi_k)\land \mathtt{HAS}(P_1, \varphi_1) \land ... \land \mathtt{HAS}(P_k, \varphi_k)$

		\end{itemize}
	\item Prefixes, suffixes, infixes
		\begin{itemize}
			\item $\mathtt{BEGIN}(P,\varphi) \eqvdef \mathtt{REP}(P,\varphi) \circ ( \varepsilon \lor \lnot(\varphi \circ \top))$
			\item $\mathtt{END}(P,\varphi) \eqvdef (\varepsilon \lor \lnot(\top \circ \varphi)) \circ \mathtt{REP}(P,\varphi)$
			\item $\mathtt{RANGE}(i, j, \varphi)$\\\hspace*{3mm}$\eqvdef \mathtt{LEN}=i-1 \circ (\varphi \land \mathtt{LEN}=i-j+1) \circ \top$
			\item $\mathtt{AT}(i, \varphi) \eqvdef \mathtt{RANGE}(i, i, \varphi)$
		\end{itemize}
	\item Counting
		\begin{itemize}
			\item $c \cdot \mathtt{COUNT}(\varphi) > d \cdot \mathtt{COUNT}(\psi)$\\ \hspace*{3mm}$\eqvdef \exists i \big(\mathtt{HAS}(>i, \varphi)\wedge \mathtt{HAS}(\leq \frac{d}{c}i, \psi)\big)$
			\item $\mathtt{COUNT}(\varphi) > \mathtt{COUNT}(\psi)$\\ \hspace*{3mm}$\eqvdef 1 \cdot \mathtt{COUNT}(\varphi) > 1 \cdot \mathtt{COUNT}(\psi)$
		\end{itemize}
	\item Unary numerical predicates
		\begin{itemize}
		 \item $\mathtt{EVEN} \eqvdef (\equiv 0 \mod 2)$
		 \item $\mathtt{ODD} \eqvdef (\equiv 1 \mod 2)$
		\end{itemize}
\end{itemize}

\subsubsection{Additional Examples}
\begin{example}
 When asked for a natural language description for the regular language given by set notation $\{a^m b^n \mid \text{$m$ even, $n$ odd}\}$ a student may provide the description  ``All strings that start with an even number of $a$'s and end with an odd number of $b$'s'' expressible in \textsc{Nile} as
			\[\red{\mathtt{BEGIN}}(\mathtt{EVEN},a)\mathbin{\red{\land}}\red{\mathtt{END}}(\mathtt{ODD},b)\]
		Comparison with a correct description ``Strings that start with an even number of $a$'s, followed only by an odd number of $b$'s''  expressible in $\textsc{Nile}$ as
		\[\red{\mathtt{REP}}(\mathtt{EVEN},a)\mathbin{\red{\circ}}\red{\mathtt{REP}}(\mathtt{ODD},b)\]
		yields, again, that consecutiveness has not been taken into account, derivable from the syntax trees of the \textsc{Nile} expressions.

\end{example}
\begin{example}
When asked for a natural language description for the context-free language given by the context-free grammar $ S \rightarrow bS \mid SS \mid aSb \mid bSa \mid \varepsilon $, a student may provide the description  ``Strings with more $b$s than $a$s'' expressible in \textsc{Nile} as
			\[\mathtt{COUNT}(b)\mathrel{\red{>}}\mathtt{COUNT}(a)\]
		Comparison with a correct description ``Strings with at least as many $b$s as $a$s'' expressible in $\textsc{Nile}$ as
		\[\mathtt{COUNT}(b)\mathrel{\red{\ge}}\mathtt{COUNT}(a)\]
		yields an explanation based on the counts of $a$s and $b$s, derivable from the \textsc{Nile} expressions.

\end{example}

\subsubsection{Proofs for Expressive Power of Nile}
\theoremNileRegular

\begin{proofsketch}
	All regular languages can be expressed in regular \Nile, because the empty language ($\bot$), words, union ($\varphi \lor \chi$), concatenation ($\varphi \circ \chi$) and the Kleene star ($\mathtt{REP}(\ge 0, \varphi)$) can be expressed.

	For the other direction, we observe that terms in regular \Nile are constants as they cannot contain variables. %
	Therefore all regular \Nile constructs can be easily translated into monadic second order logic (MSO) formulas. The result follows from the equivalence of MSO and regular languages, see e.g.\ \cite{Libkin04}.
\end{proofsketch}
We observe that because regular \Nile captures regular expressions with negations, the translation of regular \Nile expressions into deterministic finite state is $\mathsf{TOWER}$-complete \cite{StockmeyerM73,Stockmeyer74}.

\theoremNileBounded
\begin{proof}
	A bounded language $\calL \subseteq \calL(w_1^* \ldots w_m^*)$ can be written in the form $\{w_1^{x_1} \ldots w_m^{x_m}\mid (x_1,\ldots, x_m)\in S\}$ where $S$ is a semi-linear set over $\N_0^m$ \cite{Ginsburg1966a}. Thus, there are vectors $\vec a_{i,j}\in\N_0^m$ and numbers $\ell,r_i\in\N_0$ such that
	\begin{multline*}
		S= \bigcup_{i=1}^\ell\Big\{\vec a_{i,0}+\sum_{j=1}^{r_i} t_{i,j} \vec a_{i,j}\mid t_1,\ldots,t_{r_i}\in\N_0\Big\}.
	\end{multline*}
	This condition can easily be translated into a \Nile expression, so we have
	\begin{multline*}
		w=w_1^{x_1} \ldots w_m^{x_m}\in \calL \iff 
		\bar w \models\bigvee_{i=1}^\ell\Big(\exists t_1\ldots\exists t_{r_i} \\ 
		\bigcirc_{k=1}^m \mathtt{REP}\big(=a_{i,0,k}+\sum_{j=1}^{r_i} t_{i,j} a_{i,j,k},\;w_k\big)\Big)
	\end{multline*}
	where $a_{i,j,k}$ denotes the $k$th entry of the vector $\vec a_{i,j}$. The unary numerical predicate $=z$ can easily be expressed in core-\Nile by $(\ge z)\land\lnot(\ge z+1)$.
\end{proof}

\subsection{Extension by more context-free constructs}
While in our experiments only counting-based context-free constructs are used, our implementation also slightly extends core \Nile by further context-free constructs. Among others, it allows
\begin{itemize}
 \item String quantification: $\exists u \varphi$ existentially quantifies a string $u$. The semantics is defined by evaluting expressions over tuples  $(w, \beta, \gamma)$ consisting of a string $w$, an interpretation $\beta$ of free integer variables, and an interpretation $\gamma$ of free string variables. The string variables can then be used in other expressions.
 \item A reverse operator: $w \models \mathtt{REVERSE}(\varphi) \iffdef w^R \models \varphi$
\end{itemize}

This allows syntactic sugar such as:
\begin{itemize}
 \item Palindrome: \\
 $\mathtt{PALINDROME} \eqvdef \exists w [w \circ (\mathtt{LEN} \le 1) \circ \mathtt{REVERSE}(w)]$
\end{itemize}
\subsection{Additional Material for Section \ref{section:evaluation}: Details for exercises}
 We provide more details for the formal languages in our data set (see Table \ref{table:exercises}). As the exercises for our data collection were in German, we provide German and English descriptions.

 \subsubsection{Regular languages}

\begin{description}
 \item[R1:] $\{ab^n \mid n \ge 1\}$
			\begin{itemize}
			\item \textbf{German}: Die Sprache enthält genau die Wörter über dem Alphabet $\{a, b\}$, die mit $ab$ beginnen und danach nur noch beliebig viele $b$s enthalten.
			\item \textbf{English}: The language contains exactly those words over the alphabet $\{a, b\}$ that begin with $ab$ and afterwards only contain an arbitrary amount of $b$s.
			\item \textbf{Extended \Nile}: $\mathtt{ALPH}(a,b){:}\ ab \circ \mathtt{REP}(b)$%
			\item \textbf{RE}: $abb^*$
			\item \textbf{\Mosel}: $(\left|\mathtt{all}\right| \ge 2) \land (a \moselat \mathtt{fst}) \land (b \moselat \mathtt{psGt}~\mathtt{fst})$
		\end{itemize}
	\item[R2:] $\{\varepsilon\} \cup \{wa \mid w \in \{a, b\}^*\}$
		\begin{itemize}
			\item \textbf{German}: Die Sprache enthält genau die Wörter über dem Alphabet $\{a, b\}$, die nicht auf $b$ enden.
			\item \textbf{English}: The language contains exactly those words over the alphabet $\{a, b\}$ that do not end with $b$.
			\item \textbf{Extended \Nile}: $\mathtt{ALPH}(a,b){:}\ \lnot \mathtt{END}(b)$
			\item \textbf{RE}: $\varepsilon+(a+b)^*a$
			\item \textbf{\Mosel}: $\lnot \mathtt{endWt}~\lqm b \rqm$
		\end{itemize}
	\item[R3:] $\{w \in \{a, b\}^* \mid w \text{ contains } aaa \}$
		\begin{itemize}
			\item \textbf{German}: Die Sprache enthält genau die Wörter über dem Alphabet $\{a, b\}$, die $aaa$ enthalten.
			\item \textbf{English}: The language contains exactly those words over the alphabet $\{a, b\}$ that contain $aaa$.
			\item \textbf{Extended \Nile}: $\mathtt{ALPH}(a,b){:}\ \mathtt{HAS}(aaa)$
			\item \textbf{RE}: $(a+b)^*aaa(a+b)^*$
			\item \textbf{\Mosel}: $\left|\mathtt{indOf}~\lqm aaa \rqm\right| \ge 1$
		\end{itemize}
	\item[R4:] $\{a^mb^n \mid m \text{ even}, n \text{ odd} \}$
				\begin{itemize}
			\item \textbf{German}: Die Sprache enthält genau die Wörter über dem Alphabet $\{a, b\}$, die aus einem Block $a$s gefolgt von einem Block $b$s bestehen, wobei die Anzahl der $a$s gerade und die der $b$s ungerade ist.
			\item \textbf{English}: The language contains exactly those words over the alphabet $\{a, b\}$ that consist of an even-length sequence of $a$s followed by an odd-length sequence of $b$s.
			\item \textbf{Extended \Nile}: $\mathtt{ALPH}(a,b){:}\ \mathtt{REP}(\mathtt{EVEN},a) \circ \mathtt{REP}(\mathtt{ODD},b)$
			\item \textbf{RE}: $(aa)^*b(bb)^*$
			\item \textbf{\Mosel}: $\exists i.~(a \moselat \mathtt{psLt}~i) \land (b \moselat \mathtt{psGe}~i) \land (\left|\mathtt{psLt}~i\right|\moselperc 2 = 0) \land (\left|\mathtt{psGe}~i\right|\moselperc 2 = 1)$
		\end{itemize}
	\item[R5:]	$\{w \in \{a, b, c\}^* \mid \text{every } a \text{ is directly followed by } b\}$
		\begin{itemize}
			\item \textbf{German}: Die Sprache enthält genau die Wörter über dem Alphabet $\{a, b, c\}$, in denen auf jedes $a$ immer direkt ein $b$ folgt.
			\item \textbf{English}: The language contains exactly those words over the alphabet $\{a, b, c\}$ in which every occurrence of $a$ is directly followed by a $b$.
			\item \textbf{Extended \Nile}: $\mathtt{ALPH}(a,b,c){:}\ \forall i[\mathtt{AT}(=i,a) \to \mathtt{AT}(=i+1,b)]$
			\item \textbf{RE}: $(b+c+ab)^*$
			\item \textbf{\Mosel}: $\forall i.~(a \moselat i \implies b \moselat i + 1)$
		\end{itemize}
	\item[R6:] $\{w \in \{a, b, c\}^* \mid \#_b(w) = 0 \lor w \text{ ends in } aba\}$
		\begin{itemize}
			\item \textbf{German}: Die Sprache enthält genau die Wörter über dem Alphabet $\{a, b, c\}$, die entweder keine $b$s enthalten oder auf $aba$ enden.
			\item \textbf{English}: The language contains exactly those words over the alphabet $\{a, b, c\}$ that either contain no $b$s or end with $aba$.
			\item \textbf{Extended \Nile}: $\mathtt{ALPH}(a,b,c){:}\ \lnot\mathtt{HAS}(b) \lor \mathtt{END}(aba)$
			\item \textbf{RE}: $(a+c)^*+(a+b+c)^*aba$
			\item \textbf{\Mosel}: $(\left|\mathtt{indOf}~b\right| = 0) \lor (\mathtt{endWt}~\lqm aba \rqm)$
		\end{itemize}

\end{description}

 \subsubsection{Context-free languages}
\begin{description}
	\item[C1:] $\{w \in \{a, b\}^* \mid \#_b(w) \ge \#_a(w)\}$
		\begin{itemize}
			\item \textbf{German}: Die Sprache enthält genau die Wörter über dem Alphabet $\{a, b\}$ mit mindestens so vielen $b$s wie $a$s.
			\item \textbf{English}: The language contains exactly those words over the alphabet $\{a, b\}$ with at least as many $b$s as $a$s.
			\item \textbf{Extended \Nile}: $\mathtt{COUNT}(b) \ge \mathtt{COUNT}(a)$
		\end{itemize}
	\item[C2:] $\{a^nb^n \mid n \in \mathbb{N}_0 \}$
		\begin{itemize}
			\item \textbf{German}: Die Sprache enthält genau die Wörter über dem Alphabet $\{a, b\}$, die mit beliebig vielen $ a $s beginnen und danach nur noch genau so viele $b$s enthalten.
			\item \textbf{English}: The language contains exactly those words over the alphabet $\{a, b\}$ that consist of a sequence of $a$s followed by a sequence of $b$s of the same length.
			\item \textbf{Extended \Nile}: $\mathtt{ALPH}(a,b){:}\ \exists i[\mathtt{REP}(=i,a)\circ\mathtt{REP}(=i, b)]$
		\end{itemize}
	\item[C3:] $\{a^mb^n(ab)^{m+n} \mid m, n \in \mathbb{N}_0 \}$
		\begin{itemize}
			\item \textbf{German}: Die Sprache enthält genau die Wörter über dem Alphabet $\{a, b\}$, die aus einem Block $a$s, gefolgt von einem Block $b$s, gefolgt von Wiederholungen des Teilwortes $ab$ bestehen, wobei $ab$ genau so oft wiederholt wird wie die Anzahl der $a$s und $b$s zusammen.
			\item \textbf{English}: The language contains exactly those words over the alphabet $\{a, b\}$ that consist of a sequence of $a$s followed by a sequence of $b$s and finally a sequence of $ab$s, where $ab$ is repeated as often as the sum of previous repetitions of $a$s and $b$s.
			\item \textbf{Extended \Nile}: $\mathtt{ALPH}(a,b){:}\ \exists i \exists j [\mathtt{REP}(= i, a) \circ \mathtt{REP}(= j, b) \circ \mathtt{REP}(= i + j, ab)]$
		\end{itemize}
	\item[C4:] $\{wba^n \mid w \in \{a, b\}^*, n \in \mathbb{N}_0, \#_a(w) = n\}$
		\begin{itemize}
			\item \textbf{German}: Die Sprache enthält genau die Wörter über dem Alphabet $\{a, b\}$, die mindestens ein $b$ haben und nach dem letzten $b$ genau so viele $a$s haben wie davor.
			\item \textbf{English}: The language contains exactly those words over the alphabet $\{a, b\}$ that contain at least one $b$ and have as many $a$s following the last $b$ as are preceding it.
			\item \textbf{Extended \Nile}:  $\mathtt{ALPH}(a,b){:}\ \exists i[\mathtt{HAS}(= i, a) \circ b \circ \mathtt{REP}(= i, a)]$
		\end{itemize}
	\item[C5:] $\{w \in \{a, b, c\}^* \mid \#_a(w) \le 2 \cdot \#_b(w)\}$
		\begin{itemize}
			\item \textbf{German}: Die Sprache enthält genau die Wörter über dem Alphabet $\{a, b, c\}$, die höchstens doppelt so viele $a$s wie $b$s enthalten.
			\item \textbf{English}: The language contains exactly those words over the alphabet $\{a, b, c\}$ that contain at most twice as many $a$s as $b$s.
			\item \textbf{Extended \Nile}: $\mathtt{ALPH}(a,b,c):\mathtt{COUNT}(a) \le 2 \cdot \mathtt{COUNT}(b)$
		\end{itemize}
\end{description}

\subsection{Additional Material for Section \ref{section:evaluation}: \Nile Prompt details}
We give some more details for the prompts used for \Nile. As \Nile is a new language, the LLM did not encounter \Nile during training. Within the prompt, \Nile is introduced via 62 examples containing a natural language description and a corresponding \Nile construct. These examples are designed to (i) cover all core language constructs as well all syntactic sugar, and (ii) show how statements can be composed. As our student data is in German, the example descriptions are also in German. The precise prompts can be found in Appendix \ref{section:prompts}.

We note that some examples are very close to expressions that describe the formal languages R1--R6, C1--C5. This is hard to avoid if all standard constructs are illustrated; also the prompts are so extensive to support standard formal languages from educational contexts without modification. The natural language descriptions by students in our data set vary a lot: each formal language R1--R6, C1--C5 comes with between 16 and 50 natural language descriptions by students that lead to non-equivalent \Nile expressions (see Table \ref{table:exercises}).

The original \Nile prompt designed with languages R1--R3, C1, C2 performs badly for languages R5 and C3. In R5, the model reliably returned equivalent \Nile~expressions for ``every $a$ is directly followed by $b$'', but often failed to preserve the construction idea and instead defaults to regular expression-like descriptions (e.g. $\mathtt{REP}(b \lor c \lor ab)$), which is not matching the natural-language description syntactically. In C3, many syntax errors arise due to the model attempting to place two variables under a single quantifier (which our language definition does not allow, but the LLM could not know that).

Our original prompt had been designed independently of exercises (R4)-(R6) and (C3)-(C5), and thus did not include instances of relative indexing or multiple quantifiers, which would address both issues. In a realistic setting, instructors would have the ability to custom-tailor prompts for each exercise. We therefore opted to repeat both experiments with an adjusted prompt, extended by a single example each:

\begin{itemize}
	\item Example: Relative indexing
	\begin{itemize}
		\item \textbf{Natural language (German)}: Die Sprache der Wörter, in denen jedes $a$ direkt einem $b$ folgt.
		\item \textbf{Natural language (English)}: The language of strings in which every $a$ directly follows a $b$.
		\item \textbf{\Nile}: $\forall i[\mathtt{AT}(i,a) \to \mathtt{AT}(i-1,b)]$
	\end{itemize}
	\item Example: Multiple quantifiers
	\begin{itemize}
		\item \textbf{Natural language (German)}: Die Sprache der Wörter, die aus einer Folge von $a$s, dann einer Folge von $b$s und zuletzt wieder einer Folge von $a$s bestehen, wobei jede der Folgen länger ist als die vorhergehende.
		\item \textbf{Natural language (English)}: The language of strings that consist of a sequence of $a$s, followed by a sequence of $b$s, and finally another sequence of $a$s, where each sequence is longer than the previous one.
		\item \textbf{\Nile}: $\exists i \exists j[\mathtt{REP}(=i,a)\circ\mathtt{REP}((=j) \land (>i))\circ\mathtt{REP}(>j,a)]$
	\end{itemize}
\end{itemize}

The evaluation data for the original prompts and the adjusted prompts can be found in Tables \ref{table:appendix-overall+design-data} through \ref{table:appendix-accumulated-data}.

\subsection{Additional Material for Section \ref{section:evaluation}: Details for Annotation  and Evaluation}

Challenges in annotating data are (i) ambiguities in natural language descriptions of students, (ii) determinining syntactical closeness of formal representations to natural language descriptions.  We discuss how we adress these challenges next.

\subsubsection{Ambiguities}\label{section:ambiguities}

Recall that $\mathcal{I}(\mathcal{\calS})$ denotes the set of languages potentially described by a natural description $\mathcal{\calS}$. The description $\calS$ is \emph{ambiguous} if $\left|\mathcal{I}(\mathcal{S})\right| > 1$.

A natural language description $\calS$ by a student can fall in one of the following categories,  where $\mathcal{L} = \mathcal{L}(\mathcal{R})$ is the single solution language and $\tilde{\mathcal{L}}_1, \tilde{\mathcal{L}}_2, ...$ are arbitrary non-solution languages:

\begin{itemize}
	\item[(1)] $\calS$ is unambiguously correct, i.e. $\mathcal{I}(\mathcal{\calS}) = \{\mathcal{L}\}$,
	\item[(2)] $\calS$ is ambiguous to a correct solution and at least one non-solution language, i.e. $\mathcal{I}(\mathcal{\calS}) = \{\mathcal{L}, \tilde{\mathcal{L}}_1, \tilde{\mathcal{L}}_2, ...\}$,
	\item[(3)] $\calS$ is unambiguous to one non-solution language, i.e. $\mathcal{I}(\mathcal{\calS}) = \{\tilde{\mathcal{L}}_1\}$,
	\item[(4)] $\calS$ is ambiguous to multiple non-solution languages, i.e. $\mathcal{I}(\mathcal{\calS}) = \{\tilde{\mathcal{L}}_1, \tilde{\mathcal{L}}_2, ...\}$, and
	\item[(5)] $\calS$ is invalid, i.e. $\mathcal{I}(\mathcal{\calS}) = \emptyset$. More precisely, invalid submissions constitute nonsensical, contradictory or overly vague language descriptions, language descriptions made up mostly or exclusively of formal elements, operational descriptions of the representation (e.g. making reference to the states of an automaton), and text strings that are not language descriptions.
\end{itemize}

We annotated our data set with these categories; see Table \ref{table:ambiguities} for a summary.

In our evaluation, we removed all descriptions of category (5). For evaluating RQ 1, the method was expected to return ``true'' for all submissions in category (1) and return ``false'' for all submissions in categories (3) and (4). For submissions in category (2), either answer was deemed acceptable. Thus, the worst achievable performance for RQ 1 is not 0\%, but the fraction of ``true ambiguous'' submissions.

\begin{table}[t]
	\centering
  \def\arraystretch{1.1}
	\small
\begin{tabular}{rrrrrr}
Task & $\{\mathcal{L}\}$ & $\{\mathcal{L}, \tilde{\mathcal{L}}_1, (...)\}$ & $\{\tilde{\mathcal{L}}_1\}$ & $\{\tilde{\mathcal{L}}_1, \tilde{\mathcal{L}}_2, (...)\}$ & $\emptyset$ \\
\hline
R1 & 73 & 50 & 78 & 25 & 50 \\
R2 & 157 & 32 & 51 & 15 & 18 \\
R3 & 121 & 22 & 66 & 21 & 25 \\
C1 & 52 & 2 & 114 & 17 & 24 \\
C2 & 50 & 19 & 89 & 5 & 38 \\
R4 & 91 & 32 & 77 & 1 & 4 \\
R5 & 111 & 29 & 41 & 11 & 6 \\
R6 & 156 & 10 & 24 & 2 & 8 \\
C3 & 86 & 19 & 33 & 15 & 9 \\
C4 & 22 & 4 & 82 & 18 & 33 \\
C5 & 37 & 3 & 69 & 8 & 25 \\
\end{tabular}
\caption{Categorization of submissions per task by correctness, ambiguity and validity.} \label{table:ambiguities}
\end{table}

\subsubsection{Annotating Syntactic Matching}

We more precisely define our methodology for evaluating RQ 3 by clarifying our notion of ``syntactic matching''. These guidelines were formulated with the objective of being applicable universally across our representation formats (RE, CFG, \Nile), and as such rely on natural-language semantics rather than syntactical specifics of those formats. In order for a formal representation $\mathcal{R}_\mathcal{M}$ to syntactically match a natural description $\mathcal{D}$, any phrase in $\mathcal{D}$ must fall into one of three categories:

\begin{enumerate}
	\item Phrases that directly correspond to a construct in the output representation $\mathcal{R}_\mathcal{M}$ are permitted.
	\item Phrases that are expressed by a small, self-contained idiom in $\mathcal{R}_\mathcal{M}$ are permitted. For instance, the natural-language phrase ``an arbitrary sequence of symbols'' maps to the idiom $(a+b)^*$ in regular expressions when $\Sigma = \{a, b\}$.
	\item Phrases that do not correspond to any part of $\mathcal{R}_\mathcal{M}$ are permitted as long as they are obviously redundant and do not impede understanding.
\end{enumerate}

While models are permitted to make simplifications when constructing $\mathcal{R}_\mathcal{M}$ as per the third criterion, it is essential that any construct in $\mathcal{R}_\mathcal{M}$ is accounted for by a matching natural-language construct in $\mathcal{D}$. For example, the natural description $\mathcal{D} =$ ``words ending in $a$'' is not matched by the regular expression $\mathcal{R}_\mathcal{M} = (a+b)^*a$, as the subexpression $(a+b)^*$ and its concatenation to $a$ do not occur in $\mathcal{D}$ and are merely artifacts of the limited expressivity of regular expressions.

For context-free grammars, recursive productions were never considered to be self-contained idioms due to their non-local nature. This accounts for the poor performance of context-free grammars in RQ 3, as virtually all of the submitted languages required recursion to be expressed.

\clearpage

\onecolumn
\section{Detailed Evaluation Data}\label{section:appendix-detailed-evaluation-data}
\vspace{1em}
\begin{table}[!h]
	\def\arraystretch{1.2}
    \centering
    \scalebox{0.64}{%
    \begin{tabular}{l||>{\bfseries}r>{\bfseries}r|>{\bfseries}r>{\bfseries}rrrrrrrrrrr}
	\toprule
        & \multicolumn{2}{c|}{}             & \multicolumn{12}{c}{Design data} \\
        & \multicolumn{2}{c|}{\textbf{all}} & \multicolumn{2}{c}{\textbf{all}}  & \multicolumn{2}{c}{R1} & \multicolumn{2}{c}{R2} & \multicolumn{2}{c}{R3} & \multicolumn{2}{c}{C1} & \multicolumn{2}{c}{C2} \\
        & $\#$                              & $\%$                              & $\#$                   & $\%$                   & $\#$                   & $\%$                   & $\#$                    & $\%$ & $\#$ & $\%$ & $\#$ & $\%$ & $\#$ & $\%$ \\
    \midrule
	\hspace{3mm} all descriptions                                                                                                                       & 2040 & 100\%   & 1059 & 100\%   & 226 & 100\%   & 255 & 100\%    & 230 & 100\%    & 185 & 100\%    & 163 & 100\%    \\
	\hspace{9mm} describing the solution language: $\{\mathcal{L}\}$                                                                                    & 956  & & 453  &  & 73  &  & 157 &  & 121 &  & 52  &    & 50  &    \\
	\hspace{9mm} describing non-solution language(s): $\{\tilde{\mathcal{L}}_1, (...)\}$                                                                & 862  & & 481  &  & 103 &  & 66  &  & 87  &  & 131 &    & 94  &    \\
	\hspace{9mm} describing non- and solution language(s): $\{\mathcal{L}, \tilde{\mathcal{L}}_1, (...)\}$                                              & 222  & & 125  &  & 50  &  & 32  &  & 22  &  & 2   &    & 19  &    \\
	\midrule
	\multicolumn{15}{l}{\textbf{RQ 1}: How well can LLMs decide whether a natural language description fits to a formal language?}                     \\
	\midrule
	\hspace{3mm} (M1) Directly                                                                                                                          & 1956 & 95.88\% & 1019 & 96.22\% & 214 & 94.69\% & 247 & 96.86\%  & 215 & 93.48\%  & 184 & 99.46\%  & 159 & 97.55\%  \\
	\hspace{9mm} describing the solution language: $\{\mathcal{L}\}$                                                                                    & 917  & 95.92\% & 434  & 95.81\% & 67  & 91.78\% & 153 & 97.45\%  & 113 & 93.39\%  & 51  & 98.08\%  & 50  & 100.00\% \\
	\hspace{9mm} describing non-solution language(s): $\{\tilde{\mathcal{L}}_1, (...)\}$                                                                & 817  & 94.78\% & 460  & 95.63\% & 97  & 94.17\% & 62  & 93.94\%  & 80  & 91.95\%  & 131 & 100.00\% & 90  & 95.74\%  \\
	\hspace{3mm} (M2) Via classical representation (RE / CFG)                                                                                           & 1913 & 93.77\% & 1013 & 95.66\% & 224 & 99.12\% & 252 & 98.82\%  & 228 & 99.13\%  & 147 & 79.46\%  & 162 & 99.39\%  \\
	\hspace{9mm} describing the solution language: $\{\mathcal{L}\}$                                                                                    & 837  & 87.55\% & 413  & 91.17\% & 72  & 98.63\% & 156 & 99.36\%  & 119 & 98.35\%  & 16  & 30.77\%  & 50  & 100.00\% \\
	\hspace{9mm} describing non-solution language(s): $\{\tilde{\mathcal{L}}_1, (...)\}$                                                                & 854  & 99.07\% & 475  & 98.75\% & 102 & 99.03\% & 64  & 96.97\%  & 87  & 100.00\% & 129 & 98.47\%  & 93  & 98.94\%  \\
	\hspace{3mm} (M3) Via \Nile                                                                                                                         & 1950 & 95.59\% & 1046 & 98.77\% & 223 & 98.67\% & 248 & 97.25\%  & 228 & 99.13\%  & 184 & 99.46\%  & 163 & 100.00\% \\
	\hspace{9mm} describing the solution language: $\{\mathcal{L}\}$                                                                                    & 873  & 91.32\% & 443  & 97.79\% & 72  & 98.63\% & 151 & 96.18\%  & 119 & 98.35\%  & 51  & 98.08\%  & 50  & 100.00\% \\
	\hspace{9mm} describing non-solution language(s): $\{\tilde{\mathcal{L}}_1, (...)\}$                                                                & 855  & 99.19\% & 478  & 99.38\% & 101 & 98.06\% & 65  & 98.48\%  & 87  & 100.00\% & 131 & 100.00\% & 94  & 100.00\% \\
	\hspace{3mm} (M3$^*$) Via \Nile (adjusted prompts)                                                                                                  & 1991 & 97.60\% & 1046 & 98.77\% & 223 & 98.67\% & 248 & 97.25\%  & 228 & 99.13\%  & 184 & 99.46\%  & 163 & 100.00\% \\
	\hspace{9mm} describing the solution language: $\{\mathcal{L}\}$                                                                                    & 914  & 95.61\% & 443  & 97.79\% & 72  & 98.63\% & 151 & 96.18\%  & 119 & 98.35\%  & 51  & 98.08\%  & 50  & 100.00\% \\
	\hspace{9mm} describing non-solution language(s): $\{\tilde{\mathcal{L}}_1, (...)\}$                                                                & 855  & 99.19\% & 478  & 99.38\% & 101 & 98.06\% & 65  & 98.48\%  & 87  & 100.00\% & 131 & 100.00\% & 94  & 100.00\% \\
	\midrule
	\multicolumn{15}{l}{\textbf{RQ 2}: How well can LLMs construct a formal representation that semantically matches a natural language description?}  \\
	\midrule
	\hspace{3mm} (M2) Classical representation (RE / CFG)                                                                                               & 1632 & 80.00\% & 934  & 88.20\% & 213 & 94.25\% & 249 & 97.65\%  & 219 & 95.22\%  & 98  & 52.97\%  & 155 & 95.09\%  \\
	\hspace{9mm} describing the solution language: $\{\mathcal{L}\}$                                                                                    & 837  & 87.55\% & 413  & 91.17\% & 72  & 98.63\% & 156 & 99.36\%  & 119 & 98.35\%  & 16  & 30.77\%  & 50  & 100.00\% \\
	\hspace{9mm} describing non-solution language(s): $\{\tilde{\mathcal{L}}_1, (...)\}$                                                                & 602  & 69.84\% & 407  & 84.62\% & 97  & 94.17\% & 61  & 92.42\%  & 81  & 93.10\%  & 82  & 62.60\%  & 86  & 91.49\%  \\
	\hspace{9mm}  describing non- and solution language(s): $\{\mathcal{L}, \tilde{\mathcal{L}}_1, (...)\}$                                             & 193  & 86.94\% & 114  & 91.20\% & 44  & 88.00\% & 32  & 100.00\% & 19  & 86.36\%  & 0   & 0.00\%   & 19  & 100.00\% \\
	\hspace{3mm} (M3) \Nile expression                                                                                                                  & 1881 & 92.21\% & 1034 & 97.64\% & 220 & 97.35\% & 247 & 96.86\%  & 226 & 98.26\%  & 182 & 98.38\%  & 159 & 97.55\%  \\
	\hspace{9mm} describing the solution language: $\{\mathcal{L}\}$                                                                                    & 873  & 91.32\% & 443  & 97.79\% & 72  & 98.63\% & 151 & 96.18\%  & 119 & 98.35\%  & 51  & 98.08\%  & 50  & 100.00\% \\
	\hspace{9mm} describing non-solution language(s): $\{\tilde{\mathcal{L}}_1, (...)\}$                                                                & 807  & 93.62\% & 470  & 97.71\% & 99  & 96.12\% & 65  & 98.48\%  & 85  & 97.70\%  & 129 & 98.47\%  & 92  & 97.87\%  \\
	\hspace{9mm}  describing non- and solution language(s): $\{\mathcal{L}, \tilde{\mathcal{L}}_1, (...)\}$                                             & 201  & 90.54\% & 121  & 96.80\% & 49  & 98.00\% & 31  & 96.88\%  & 22  & 100.00\% & 2   & 100.00\% & 17  & 89.47\%  \\
	\hspace{3mm} (M3$^*$) \Nile expression (adjusted prompts)                                                                                           & 1939 & 95.05\% & 1034 & 97.64\% & 220 & 97.35\% & 247 & 96.86\%  & 226 & 98.26\%  & 182 & 98.38\%  & 159 & 97.55\%  \\
	\hspace{9mm} describing the solution language: $\{\mathcal{L}\}$                                                                                    & 914  & 95.61\% & 443  & 97.79\% & 72  & 98.63\% & 151 & 96.18\%  & 119 & 98.35\%  & 51  & 98.08\%  & 50  & 100.00\% \\
	\hspace{9mm} describing non-solution language(s): $\{\tilde{\mathcal{L}}_1, (...)\}$                                                                & 816  & 94.66\% & 470  & 97.71\% & 99  & 96.12\% & 65  & 98.48\%  & 85  & 97.70\%  & 129 & 98.47\%  & 92  & 97.87\%  \\
	\hspace{9mm}  describing non- and solution language(s): $\{\mathcal{L}, \tilde{\mathcal{L}}_1, (...)\}$                                             & 209  & 94.14\% & 121  & 96.80\% & 49  & 98.00\% & 31  & 96.88\%  & 22  & 100.00\% & 2   & 100.00\% & 17  & 89.47\%  \\
	\midrule
	\multicolumn{15}{l}{\textbf{RQ 3}: How well can LLMs construct a formal representation that syntactically matches a natural language description?} \\
	\midrule
	\hspace{3mm} (M2) Classical representation (RE / CFG)                                                                                               & 493  & 24.17\% & 388  & 36.64\% & 144 & 63.72\% & 186 & 72.94\%  & 55  & 23.91\%  & 3   & 1.62\%   & 0   & 0.00\%   \\
	\hspace{9mm} describing the solution language: $\{\mathcal{L}\}$                                                                                    & 297  & 31.07\% & 219  & 48.34\% & 72  & 98.63\% & 123 & 78.34\%  & 24  & 19.83\%  & 0   & 0.00\%   & 0   & 0.00\%   \\
	\hspace{9mm} describing non-solution language(s): $\{\tilde{\mathcal{L}}_1, (...)\}$                                                                & 135  & 15.66\% & 121  & 25.16\% & 52  & 50.49\% & 39  & 59.09\%  & 27  & 31.03\%  & 3   & 2.29\%   & 0   & 0.00\%   \\
	\hspace{9mm}  describing non- and solution language(s): $\{\mathcal{L}, \tilde{\mathcal{L}}_1, (...)\}$                                             & 61   & 27.48\% & 48   & 38.40\% & 20  & 40.00\% & 24  & 75.00\%  & 4   & 18.18\%  & 0   & 0.00\%   & 0   & 0.00\%   \\
	\hspace{3mm} (M3) \Nile expression                                                                                                                  & 1742 & 85.39\% & 1013 & 95.66\% & 219 & 96.90\% & 240 & 94.12\%  & 220 & 95.65\%  & 177 & 95.68\%  & 157 & 96.32\%  \\
	\hspace{9mm} describing the solution language: $\{\mathcal{L}\}$                                                                                    & 800  & 83.68\% & 440  & 97.13\% & 71  & 97.26\% & 150 & 95.54\%  & 119 & 98.35\%  & 51  & 98.08\%  & 49  & 98.00\%  \\
	\hspace{9mm} describing non-solution language(s): $\{\tilde{\mathcal{L}}_1, (...)\}$                                                                & 768  & 89.10\% & 459  & 95.43\% & 99  & 96.12\% & 62  & 93.94\%  & 82  & 94.25\%  & 124 & 94.66\%  & 92  & 97.87\%  \\
	\hspace{9mm}  describing non- and solution language(s): $\{\mathcal{L}, \tilde{\mathcal{L}}_1, (...)\}$                                             & 174  & 78.38\% & 114  & 91.20\% & 49  & 98.00\% & 28  & 87.50\%  & 19  & 86.36\%  & 2   & 100.00\% & 16  & 84.21\%  \\
	\hspace{3mm} (M3$^*$) \Nile expression (adjusted prompts)                                                                                           & 1845 & 90.44\% & 1013 & 95.66\% & 219 & 96.90\% & 240 & 94.12\%  & 220 & 95.65\%  & 177 & 95.68\%  & 157 & 96.32\%  \\
	\hspace{9mm} describing the solution language: $\{\mathcal{L}\}$                                                                                    & 875  & 91.53\% & 440  & 97.13\% & 71  & 97.26\% & 150 & 95.54\%  & 119 & 98.35\%  & 51  & 98.08\%  & 49  & 98.00\%  \\
	\hspace{9mm} describing non-solution language(s): $\{\tilde{\mathcal{L}}_1, (...)\}$                                                                & 782  & 90.72\% & 459  & 95.43\% & 99  & 96.12\% & 62  & 93.94\%  & 82  & 94.25\%  & 124 & 94.66\%  & 92  & 97.87\%  \\
	\hspace{9mm}  describing non- and solution language(s): $\{\mathcal{L}, \tilde{\mathcal{L}}_1, (...)\}$                                             & 188  & 84.68\% & 114  & 91.20\% & 49  & 98.00\% & 28  & 87.50\%  & 19  & 86.36\%  & 2   & 100.00\% & 16  & 84.21\%  \\
	\bottomrule
    \end{tabular}
    }
	\caption{Detailed data for the evaluation of RQs 1 to 3. Overall performance, performance for regular languages R1 -- R3 and performance for context-free languages C1 -- C2 is shown. For RQs 2 and 3, the LLM was asked to construct a semantically and syntactically matching formal representation (RE, CFG, \Nile). For RQ~1, in addition to formally checking equivalence of the LLM output to the solution language, the LLM was also asked to judge equivalence directly. Only design data (languages R1 -- R3, C1, C2) has been used for conceptual and prompt design. The prompt (M3$^*$) denotes a posthoc extension of prompt (M3) by two additional examples for non-covered \Nile idioms. This adjustment significantly increased the performance for R5 and C3. We further distinguish descriptions according to whether they unambiguously describe the solution language, unambiguously describe non-solution languages, or can describe a solution and a non-solution language, corresponding to Category (1), Categories (3) -- (4), and Category (2) respectively, as defined in Section \ref{section:ambiguities}. Since ambiguous submissions achieve a RQ~1 performance of 100\% by definition, the corresponding rows have been omitted. }\label{table:appendix-overall+design-data}
\end{table}

\newcommand{\fiiill}{&&&&&&&&&&&&&&&&}
\newcommand{\tableindent}{\hspace{1em}}
\newcommand{\tablesubindent}{\hspace{2.5em}}

\begin{table}[!h]
    \centering
	\def\arraystretch{1.2}
    \scalebox{0.645}{%
    \begin{tabular}{l||>{\bfseries}r>{\bfseries}rrrrrrrrrrrrr}%
	\toprule
          & \multicolumn{14}{c}{Validation data} \\
          & \multicolumn{2}{c}{\textbf{all}} & \multicolumn{2}{c}{R4}                & \multicolumn{2}{c}{R5} & \multicolumn{2}{c}{R6} & \multicolumn{2}{c}{C3} & \multicolumn{2}{c}{C4} & \multicolumn{2}{c}{C5} \\
          & $\#$                               & $\%$                                    & $\#$                     & $\%$                     & $\#$                     & $\%$                     & $\#$                       & $\%$       & $\#$  & $\%$       & $\#$  & $\%$       & $\#$  & $\%$       \\
	\midrule
	\hspace{3mm} all descriptions                                                                                                                       & 981 & 100\%   & 201 & 100\%    & 192 & 100\%   & 192 & 100\%    & 153 & 100\%    & 126 & 100\%    & 117 & 100\%    \\
	\hspace{9mm} describing the solution language: $\{\mathcal{L}\}$                                                                                    & 503 &   & 91  &    & 111 &   & 156 &    & 86  &    & 22  &    & 37  &    \\
	\hspace{9mm} describing non-solution language(s): $\{\tilde{\mathcal{L}}_1, (...)\}$                                                                & 381 &   & 78  &    & 52  &   & 26  &    & 48  &    & 100 &    & 77  &    \\
	\hspace{9mm}  describing non- and solution language(s): $\{\mathcal{L}, \tilde{\mathcal{L}}_1, (...)\}$                                             & 97  &   & 32  &    & 29  &   & 10  &    & 19  &    & 4   &    & 3   &    \\
	\midrule
	\multicolumn{15}{l}{\textbf{RQ 1}: How well can LLMs decide whether a natural language description fits to a formal language?}                     \\
	\midrule
	\hspace{3mm} (M1) Directly                                                                                                                          & 937 & 95.51\% & 196 & 97.51\%  & 179 & 93.23\% & 189 & 98.44\%  & 148 & 96.73\%  & 112 & 88.89\%  & 113 & 96.58\%  \\
	\hspace{9mm} describing the solution language: $\{\mathcal{L}\}$                                                                                    & 483 & 96.02\% & 91  & 100.00\% & 100 & 90.09\% & 153 & 98.08\%  & 83  & 96.51\%  & 21  & 95.45\%  & 35  & 94.59\%  \\
	\hspace{9mm} describing non-solution language(s): $\{\tilde{\mathcal{L}}_1, (...)\}$                                                                & 357 & 93.70\% & 73  & 93.59\%  & 50  & 96.15\% & 26  & 100.00\% & 46  & 95.83\%  & 87  & 87.00\%  & 75  & 97.40\%  \\
	\hspace{3mm} (M2) Via classical representation (RE / CFG)                                                                                           & 900 & 91.74\% & 201 & 100.00\% & 187 & 97.40\% & 184 & 95.83\%  & 126 & 82.35\%  & 122 & 96.83\%  & 80  & 68.38\%  \\
	\hspace{9mm} describing the solution language: $\{\mathcal{L}\}$                                                                                    & 424 & 84.29\% & 91  & 100.00\% & 108 & 97.30\% & 148 & 94.87\%  & 59  & 68.60\%  & 18  & 81.82\%  & 0   & 0.00\%   \\
	\hspace{9mm} describing non-solution language(s): $\{\tilde{\mathcal{L}}_1, (...)\}$                                                                & 379 & 99.48\% & 78  & 100.00\% & 50  & 96.15\% & 26  & 100.00\% & 48  & 100.00\% & 100 & 100.00\% & 77  & 100.00\% \\
	\hspace{3mm} (M3) Via \Nile                                                                                                                         & 904 & 92.15\% & 196 & 97.51\%  & 188 & 97.92\% & 190 & 98.96\%  & 100 & 65.36\%  & 114 & 90.48\%  & 116 & 99.15\%  \\
	\hspace{9mm} describing the solution language: $\{\mathcal{L}\}$                                                                                    & 430 & 85.49\% & 88  & 96.70\%  & 108 & 97.30\% & 154 & 98.72\%  & 33  & 38.37\%  & 10  & 45.45\%  & 37  & 100.00\% \\
	\hspace{9mm} describing non-solution language(s): $\{\tilde{\mathcal{L}}_1, (...)\}$                                                                & 377 & 98.95\% & 76  & 97.44\%  & 51  & 98.08\% & 26  & 100.00\% & 48  & 100.00\% & 100 & 100.00\% & 76  & 98.70\%  \\
	\hspace{3mm} (M3$^*$) Via \Nile (adjusted prompts)                                                                                                  & 945 & 96.33\% & 196 & 97.51\%  & 188 & 97.92\% & 190 & 98.96\%  & 141 & 92.16\%  & 114 & 90.48\%  & 116 & 99.15\%  \\
	\hspace{9mm} describing the solution language: $\{\mathcal{L}\}$                                                                                    & 471 & 93.64\% & 88  & 96.70\%  & 108 & 97.30\% & 154 & 98.72\%  & 74  & 86.05\%  & 10  & 45.45\%  & 37  & 100.00\% \\
	\hspace{9mm} describing non-solution language(s): $\{\tilde{\mathcal{L}}_1, (...)\}$                                                                & 377 & 98.95\% & 76  & 97.44\%  & 51  & 98.08\% & 26  & 100.00\% & 48  & 100.00\% & 100 & 100.00\% & 76  & 98.70\%  \\
	\midrule
	\multicolumn{15}{l}{\textbf{RQ 2}: How well can LLMs construct a formal representation that semantically matches a natural language description?}  \\
	\midrule
	\hspace{3mm} (M2) Classical representation (RE / CFG)                                                                                               & 698 & 71.15\% & 125 & 62.19\%  & 177 & 92.19\% & 182 & 94.79\%  & 99  & 64.71\%  & 86  & 68.25\%  & 29  & 24.79\%  \\
	\hspace{9mm} describing the solution language: $\{\mathcal{L}\}$                                                                                    & 424 & 84.29\% & 91  & 100.00\% & 108 & 97.30\% & 148 & 94.87\%  & 59  & 68.60\%  & 18  & 81.82\%  & 0   & 0.00\%   \\
	\hspace{9mm} describing non-solution language(s): $\{\tilde{\mathcal{L}}_1, (...)\}$                                                                & 195 & 51.18\% & 5   & 6.41\%   & 42  & 80.77\% & 25  & 96.15\%  & 29  & 60.42\%  & 65  & 65.00\%  & 29  & 37.66\%  \\
	\hspace{9mm}  describing non- and solution language(s): $\{\mathcal{L}, \tilde{\mathcal{L}}_1, (...)\}$                                             & 79  & 81.44\% & 29  & 90.63\%  & 27  & 93.10\% & 9   & 90.00\%  & 11  & 57.89\%  & 3   & 75.00\%  & 0   & 0.00\%   \\
	\hspace{3mm} (M3) \Nile expression                                                                                                                  & 847 & 86.34\% & 192 & 95.52\%  & 187 & 97.40\% & 187 & 97.40\%  & 74  & 48.37\%  & 102 & 80.95\%  & 105 & 89.74\%  \\
	\hspace{9mm} describing the solution language: $\{\mathcal{L}\}$                                                                                    & 430 & 85.49\% & 88  & 96.70\%  & 108 & 97.30\% & 154 & 98.72\%  & 33  & 38.37\%  & 10  & 45.45\%  & 37  & 100.00\% \\
	\hspace{9mm} describing non-solution language(s): $\{\tilde{\mathcal{L}}_1, (...)\}$                                                                & 337 & 88.45\% & 75  & 96.15\%  & 51  & 98.08\% & 25  & 96.15\%  & 32  & 66.67\%  & 88  & 88.00\%  & 66  & 85.71\%  \\
	\hspace{9mm}  describing non- and solution language(s): $\{\mathcal{L}, \tilde{\mathcal{L}}_1, (...)\}$                                             & 80  & 82.47\% & 29  & 90.63\%  & 28  & 96.55\% & 8   & 80.00\%  & 9   & 47.37\%  & 4   & 100.00\% & 2   & 66.67\%  \\
	\hspace{3mm} (M3$^*$) \Nile expression (adjusted prompts)                                                                                           & 905 & 92.25\% & 192 & 95.52\%  & 186 & 96.88\% & 187 & 97.40\%  & 133 & 86.93\%  & 102 & 80.95\%  & 105 & 89.74\%  \\
	\hspace{9mm} describing the solution language: $\{\mathcal{L}\}$                                                                                    & 471 & 93.64\% & 88  & 96.70\%  & 108 & 97.30\% & 154 & 98.72\%  & 74  & 86.05\%  & 10  & 45.45\%  & 37  & 100.00\% \\
	\hspace{9mm} describing non-solution language(s): $\{\tilde{\mathcal{L}}_1, (...)\}$                                                                & 346 & 90.81\% & 75  & 96.15\%  & 51  & 98.08\% & 25  & 96.15\%  & 41  & 85.42\%  & 88  & 88.00\%  & 66  & 85.71\%  \\
	\hspace{9mm}  describing non- and solution language(s): $\{\mathcal{L}, \tilde{\mathcal{L}}_1, (...)\}$                                             & 88  & 90.72\% & 29  & 90.63\%  & 27  & 93.10\% & 8   & 80.00\%  & 18  & 94.74\%  & 4   & 100.00\% & 2   & 66.67\%  \\
	\midrule
	\multicolumn{15}{l}{\textbf{RQ 3}: How well can LLMs construct a formal representation that syntactically matches a natural language description?} \\
	\midrule
	\hspace{3mm} (M2) Classical representation (RE / CFG)                                                                                               & 105 & 10.70\% & 0   & 0.00\%   & 44  & 22.92\% & 60  & 31.25\%  & 0   & 0.00\%   & 1   & 0.79\%   & 0   & 0.00\%   \\
	\hspace{9mm} describing the solution language: $\{\mathcal{L}\}$                                                                                    & 78  & 15.51\% & 0   & 0.00\%   & 26  & 23.42\% & 52  & 33.33\%  & 0   & 0.00\%   & 0   & 0.00\%   & 0   & 0.00\%   \\
	\hspace{9mm} describing non-solution language(s): $\{\tilde{\mathcal{L}}_1, (...)\}$                                                                & 14  & 3.67\%  & 0   & 0.00\%   & 10  & 19.23\% & 3   & 11.54\%  & 0   & 0.00\%   & 1   & 1.00\%   & 0   & 0.00\%   \\
	\hspace{9mm}  describing non- and solution language(s): $\{\mathcal{L}, \tilde{\mathcal{L}}_1, (...)\}$                                             & 13  & 13.40\% & 0   & 0.00\%   & 8   & 27.59\% & 5   & 50.00\%  & 0   & 0.00\%   & 0   & 0.00\%   & 0   & 0.00\%   \\
	\hspace{3mm} (M3) \Nile expression                                                                                                                  & 729 & 74.31\% & 178 & 88.56\%  & 113 & 58.85\% & 183 & 95.31\%  & 67  & 43.79\%  & 88  & 69.84\%  & 100 & 85.47\%  \\
	\hspace{9mm} describing the solution language: $\{\mathcal{L}\}$                                                                                    & 360 & 71.57\% & 81  & 89.01\%  & 55  & 49.55\% & 151 & 96.79\%  & 31  & 36.05\%  & 6   & 27.27\%  & 36  & 97.30\%  \\
	\hspace{9mm} describing non-solution language(s): $\{\tilde{\mathcal{L}}_1, (...)\}$                                                                & 309 & 81.10\% & 71  & 91.03\%  & 45  & 86.54\% & 25  & 96.15\%  & 27  & 56.25\%  & 79  & 79.00\%  & 62  & 80.52\%  \\
	\hspace{9mm}  describing non- and solution language(s): $\{\mathcal{L}, \tilde{\mathcal{L}}_1, (...)\}$                                             & 60  & 61.86\% & 26  & 81.25\%  & 13  & 44.83\% & 7   & 70.00\%  & 9   & 47.37\%  & 3   & 75.00\%  & 2   & 66.67\%  \\
	\hspace{3mm} (M3$^*$) \Nile expression (adjusted prompts)                                                                                           & 832 & 84.81\% & 178 & 88.56\%  & 153 & 79.69\% & 183 & 95.31\%  & 130 & 84.97\%  & 88  & 69.84\%  & 100 & 85.47\%  \\
	\hspace{9mm} describing the solution language: $\{\mathcal{L}\}$                                                                                    & 435 & 86.48\% & 81  & 89.01\%  & 87  & 78.38\% & 151 & 96.79\%  & 74  & 86.05\%  & 6   & 27.27\%  & 36  & 97.30\%  \\
	\hspace{9mm} describing non-solution language(s): $\{\tilde{\mathcal{L}}_1, (...)\}$                                                                & 323 & 84.78\% & 71  & 91.03\%  & 48  & 92.31\% & 25  & 96.15\%  & 38  & 79.17\%  & 79  & 79.00\%  & 62  & 80.52\%  \\
	\hspace{9mm}  describing non- and solution language(s): $\{\mathcal{L}, \tilde{\mathcal{L}}_1, (...)\}$                                             & 74  & 76.29\% & 26  & 81.25\%  & 18  & 62.07\% & 7   & 70.00\%  & 18  & 94.74\%  & 3   & 75.00\%  & 2   & 66.67\%  \\
	\bottomrule
    \end{tabular}
    }
	\caption{Detailed data for the evaluation of RQs 1 to 3. Performance for regular languages R4 -- R6 and for context-free languages C3 -- C5 is shown. For RQs 2 and 3, the LLM was asked to construct a semantically and syntactically matching formal representation (RE, CFG, \Nile). For RQ~1, in addition to formally checking equivalence of the LLM output to the solution language, the LLM was also asked to judge equivalence directly. Only design data (languages R1 -- R3, C1, C2, see \cref{table:appendix-overall+design-data}) has been used for conceptual and prompt design. The prompt (M3$^*$) denotes a posthoc extension of prompt (M3) by two additional examples for non-covered \Nile idioms. This adjustment significantly increased the performance for R5 and C3. We further distinguish descriptions according to whether they unambiguously describe the solution language, unambiguously describe non-solution languages, or can describe a solution and a non-solution language, corresponding to Category (1), Categories (3) -- (4), and Category (2) respectively, as defined in Section \ref{section:ambiguities}. Since ambiguous submissions achieve a RQ~1 performance of 100\% by definition, the corresponding rows have been omitted. }\label{table:appendix-validation-data}
\end{table}

\begin{table}[!h]
    \centering
	\def\arraystretch{1.2}
    \scalebox{0.715}{%
    \begin{tabular}{l||>{\bfseries}r>{\bfseries}rrrrr|>{\bfseries}r>{\bfseries}rrrrr}
	\toprule
        & \multicolumn{6}{c|}{Regular languages} & \multicolumn{6}{c}{Context-free languages} \\
        & \multicolumn{2}{c}{\textbf{all}} & \multicolumn{2}{c}{Design data} & \multicolumn{2}{c|}{Validation data} & \multicolumn{2}{c}{\textbf{all}} & \multicolumn{2}{c}{Design data} & \multicolumn{2}{c}{Validation data} \\
        & $\#$   & $\%$      & $\#$  & $\%$      & $\#$  & $\%$      & $\#$  & $\%$      & $\#$  & $\%$       & $\#$  & $\%$       \\
	\midrule
	\hspace{3mm} all descriptions                                                                                                                       & 1296 & 100\%   & 711 & 100\%   & 585 & 100\%   & 744 & 100\%   & 348 & 100\%    & 396 & 100\%    \\
	\hspace{9mm} describing the solution language: $\{\mathcal{L}\}$                                                                                    & 709  &   & 351 &   & 358 &   & 247 &   & 102 &    & 145 &    \\
	\hspace{9mm} describing non-solution language(s): $\{\tilde{\mathcal{L}}_1, (...)\}$                                                                & 412  &   & 256 &   & 156 &   & 450 &   & 225 &    & 225 &    \\
	\hspace{9mm}  describing non- and solution language(s): $\{\mathcal{L}, \tilde{\mathcal{L}}_1, (...)\}$                                             & 175  &   & 104 &   & 71  &   & 47  &   & 21  &    & 26  &    \\
	\midrule
	\multicolumn{12}{l}{\textbf{RQ 1}: How well can LLMs decide whether a natural language description fits to a formal language?}                     \\
	\midrule
	\hspace{3mm} (M1) Directly                                                                                                                          & 1240 & 95.68\% & 676 & 95.08\% & 564 & 96.41\% & 716 & 96.24\% & 343 & 98.56\%  & 373 & 94.19\%  \\
	\hspace{9mm} describing the solution language: $\{\mathcal{L}\}$                                                                                    & 677  & 95.49\% & 333 & 94.87\% & 344 & 96.09\% & 240 & 97.17\% & 101 & 99.02\%  & 139 & 95.86\%  \\
	\hspace{9mm} describing non-solution language(s): $\{\tilde{\mathcal{L}}_1, (...)\}$                                                                & 388  & 94.17\% & 239 & 93.36\% & 149 & 95.51\% & 429 & 95.33\% & 221 & 98.22\%  & 208 & 92.44\%  \\
	\hspace{3mm} (M2) Via classical representation (RE / CFG)                                                                                           & 1276 & 98.46\% & 704 & 99.02\% & 572 & 97.78\% & 637 & 85.62\% & 309 & 88.79\%  & 328 & 82.83\%  \\
	\hspace{9mm} describing the solution language: $\{\mathcal{L}\}$                                                                                    & 694  & 97.88\% & 347 & 98.86\% & 347 & 96.93\% & 143 & 57.89\% & 66  & 64.71\%  & 77  & 53.10\%  \\
	\hspace{9mm} describing non-solution language(s): $\{\tilde{\mathcal{L}}_1, (...)\}$                                                                & 407  & 98.79\% & 253 & 98.83\% & 154 & 98.72\% & 447 & 99.33\% & 222 & 98.67\%  & 225 & 100.00\% \\
	\hspace{3mm} (M3) Via \Nile                                                                                                                         & 1273 & 98.23\% & 699 & 98.31\% & 574 & 98.12\% & 677 & 90.99\% & 347 & 99.71\%  & 330 & 83.33\%  \\
	\hspace{9mm} describing the solution language: $\{\mathcal{L}\}$                                                                                    & 692  & 97.60\% & 342 & 97.44\% & 350 & 97.77\% & 181 & 73.28\% & 101 & 99.02\%  & 80  & 55.17\%  \\
	\hspace{9mm} describing non-solution language(s): $\{\tilde{\mathcal{L}}_1, (...)\}$                                                                & 406  & 98.54\% & 253 & 98.83\% & 153 & 98.08\% & 449 & 99.78\% & 225 & 100.00\% & 224 & 99.56\%  \\
	\hspace{3mm} (M3$^*$) Via \Nile (adjusted prompts)                                                                                                  & 1273 & 98.23\% & 699 & 98.31\% & 574 & 98.12\% & 718 & 96.51\% & 347 & 99.71\%  & 371 & 93.69\%  \\
	\hspace{9mm} describing the solution language: $\{\mathcal{L}\}$                                                                                    & 692  & 97.60\% & 342 & 97.44\% & 350 & 97.77\% & 222 & 89.88\% & 101 & 99.02\%  & 121 & 83.45\%  \\
	\hspace{9mm} describing non-solution language(s): $\{\tilde{\mathcal{L}}_1, (...)\}$                                                                & 406  & 98.54\% & 253 & 98.83\% & 153 & 98.08\% & 449 & 99.78\% & 225 & 100.00\% & 224 & 99.56\%  \\
	\midrule
	\multicolumn{12}{l}{\textbf{RQ 2}: How well can LLMs construct a formal representation that semantically matches a natural language description?}  \\
	\midrule
	\hspace{3mm} (M2) Classical representation (RE / CFG)                                                                                               & 1165 & 89.89\% & 681 & 95.78\% & 484 & 82.74\% & 467 & 62.77\% & 253 & 72.70\%  & 214 & 54.04\%  \\
	\hspace{9mm} describing the solution language: $\{\mathcal{L}\}$                                                                                    & 694  & 97.88\% & 347 & 98.86\% & 347 & 96.93\% & 143 & 57.89\% & 66  & 64.71\%  & 77  & 53.10\%  \\
	\hspace{9mm} describing non-solution language(s): $\{\tilde{\mathcal{L}}_1, (...)\}$                                                                & 311  & 75.49\% & 239 & 93.36\% & 72  & 46.15\% & 291 & 64.67\% & 168 & 74.67\%  & 123 & 54.67\%  \\
	\hspace{9mm}  describing non- and solution language(s): $\{\mathcal{L}, \tilde{\mathcal{L}}_1, (...)\}$                                             & 160  & 91.43\% & 95  & 91.35\% & 65  & 91.55\% & 33  & 70.21\% & 19  & 90.48\%  & 14  & 53.85\%  \\
	\hspace{3mm} (M3) \Nile expression                                                                                                                  & 1259 & 97.15\% & 693 & 97.47\% & 566 & 96.75\% & 622 & 83.60\% & 341 & 97.99\%  & 281 & 70.96\%  \\
	\hspace{9mm} describing the solution language: $\{\mathcal{L}\}$                                                                                    & 692  & 97.60\% & 342 & 97.44\% & 350 & 97.77\% & 181 & 73.28\% & 101 & 99.02\%  & 80  & 55.17\%  \\
	\hspace{9mm} describing non-solution language(s): $\{\tilde{\mathcal{L}}_1, (...)\}$                                                                & 400  & 97.09\% & 249 & 97.27\% & 151 & 96.79\% & 407 & 90.44\% & 221 & 98.22\%  & 186 & 82.67\%  \\
	\hspace{9mm}  describing non- and solution language(s): $\{\mathcal{L}, \tilde{\mathcal{L}}_1, (...)\}$                                             & 167  & 95.43\% & 102 & 98.08\% & 65  & 91.55\% & 34  & 72.34\% & 19  & 90.48\%  & 15  & 57.69\%  \\
	\hspace{3mm} (M3$^*$) \Nile expression (adjusted prompts)                                                                                           & 1258 & 97.07\% & 693 & 97.47\% & 565 & 96.58\% & 681 & 91.53\% & 341 & 97.99\%  & 340 & 85.86\%  \\
	\hspace{9mm} describing the solution language: $\{\mathcal{L}\}$                                                                                    & 692  & 97.60\% & 342 & 97.44\% & 350 & 97.77\% & 222 & 89.88\% & 101 & 99.02\%  & 121 & 83.45\%  \\
	\hspace{9mm} describing non-solution language(s): $\{\tilde{\mathcal{L}}_1, (...)\}$                                                                & 400  & 97.09\% & 249 & 97.27\% & 151 & 96.79\% & 416 & 92.44\% & 221 & 98.22\%  & 195 & 86.67\%  \\
	\hspace{9mm}  describing non- and solution language(s): $\{\mathcal{L}, \tilde{\mathcal{L}}_1, (...)\}$                                             & 166  & 94.86\% & 102 & 98.08\% & 64  & 90.14\% & 43  & 91.49\% & 19  & 90.48\%  & 24  & 92.31\%  \\
	\midrule
	\multicolumn{12}{l}{\textbf{RQ 3}: How well can LLMs construct a formal representation that syntactically matches a natural language description?} \\
	\midrule
	\hspace{3mm} (M2) Classical representation (RE / CFG)                                                                                               & 489  & 37.73\% & 385 & 54.15\% & 104 & 17.78\% & 4   & 0.54\%  & 3   & 0.86\%   & 1   & 0.25\%   \\
	\hspace{9mm} describing the solution language: $\{\mathcal{L}\}$                                                                                    & 297  & 41.89\% & 219 & 62.39\% & 78  & 21.79\% & 0   & 0.00\%  & 0   & 0.00\%   & 0   & 0.00\%   \\
	\hspace{9mm} describing non-solution language(s): $\{\tilde{\mathcal{L}}_1, (...)\}$                                                                & 131  & 31.80\% & 118 & 46.09\% & 13  & 8.33\%  & 4   & 0.89\%  & 3   & 1.33\%   & 1   & 0.44\%   \\
	\hspace{9mm}  describing non- and solution language(s): $\{\mathcal{L}, \tilde{\mathcal{L}}_1, (...)\}$                                             & 61   & 34.86\% & 48  & 46.15\% & 13  & 18.31\% & 0   & 0.00\%  & 0   & 0.00\%   & 0   & 0.00\%   \\
	\hspace{3mm} (M3) \Nile expression                                                                                                                  & 1153 & 88.97\% & 679 & 95.50\% & 474 & 81.03\% & 589 & 79.17\% & 334 & 95.98\%  & 255 & 64.39\%  \\
	\hspace{9mm} describing the solution language: $\{\mathcal{L}\}$                                                                                    & 627  & 88.43\% & 340 & 96.87\% & 287 & 80.17\% & 173 & 70.04\% & 100 & 98.04\%  & 73  & 50.34\%  \\
	\hspace{9mm} describing non-solution language(s): $\{\tilde{\mathcal{L}}_1, (...)\}$                                                                & 384  & 93.20\% & 243 & 94.92\% & 141 & 90.38\% & 384 & 85.33\% & 216 & 96.00\%  & 168 & 74.67\%  \\
	\hspace{9mm}  describing non- and solution language(s): $\{\mathcal{L}, \tilde{\mathcal{L}}_1, (...)\}$                                             & 142  & 81.14\% & 96  & 92.31\% & 46  & 64.79\% & 32  & 68.09\% & 18  & 85.71\%  & 14  & 53.85\%  \\
	\hspace{3mm} (M3$^*$) \Nile expression (adjusted prompts)                                                                                           & 1193 & 92.05\% & 679 & 95.50\% & 514 & 87.86\% & 652 & 87.63\% & 334 & 95.98\%  & 318 & 80.30\%  \\
	\hspace{9mm} describing the solution language: $\{\mathcal{L}\}$                                                                                    & 659  & 92.95\% & 340 & 96.87\% & 319 & 89.11\% & 216 & 87.45\% & 100 & 98.04\%  & 116 & 80.00\%  \\
	\hspace{9mm} describing non-solution language(s): $\{\tilde{\mathcal{L}}_1, (...)\}$                                                                & 387  & 93.93\% & 243 & 94.92\% & 144 & 92.31\% & 395 & 87.78\% & 216 & 96.00\%  & 179 & 79.56\%  \\
	\hspace{9mm}  describing non- and solution language(s): $\{\mathcal{L}, \tilde{\mathcal{L}}_1, (...)\}$                                             & 147  & 84.00\% & 96  & 92.31\% & 51  & 71.83\% & 41  & 87.23\% & 18  & 85.71\%  & 23  & 88.46\%  \\
	\bottomrule
    \end{tabular}
    }
	\caption{Detailed data for the evaluation of RQs 1 to 3. Accumulated performance for regular languages and context-free languages (overall and split into design data and validation data as indicated in \cref{table:appendix-overall+design-data,table:appendix-validation-data}) is shown. For RQs 2 and 3, the LLM was asked to construct a semantically and syntactically matching formal representation (RE, CFG, \Nile). For RQ~1, in addition to formally checking equivalence of the LLM output to the solution language, the LLM was also asked to judge equivalence directly. Only design data (languages R1 -- R3, C1, C2, see \cref{table:appendix-overall+design-data}) has been used for conceptual and prompt design. The prompt (M3$^*$) denotes a posthoc extension of prompt (M3) by two additional examples for non-covered \Nile idioms. This adjustment significantly increased the performance for R5 and C3. We further distinguish descriptions according to whether they unambiguously describe the solution language, unambiguously describe non-solution languages, or can describe a solution and a non-solution language, corresponding to Category (1), Categories (3) -- (4), and Category (2) respectively, as defined in Section \ref{section:ambiguities}. Since ambiguous submissions achieve a RQ~1 performance of 100\% by definition, the corresponding rows have been omitted. }\label{table:appendix-accumulated-data}
\end{table}

\clearpage
\section{Prompts}\label{section:prompts}

In the following prompts, \texttt{\%ALPH\%} is replaced by \texttt{\{a, b, c\}} for tasks (R5), (R6) and (C5), and \texttt{\{a, b\}} for all other tasks.

\subsection{M1: Deciding equivalence}

\begin{lstlisting}
# Instructions

* You are given a German description of a formal language submitted by a student, together with a German description of the language they were attempting to describe.
* The descriptions may match even if they don't use the same wording.
* Think about your solution. Then return the exact string "true" if the student's description semantically matches the solution; otherwise, return the exact string "false".

# Examples

Input: <student>Wörter über dem Alphabet {a, b}, die drei as enthalten.</student><solution>Wörter über dem Alphabet {a, b}, die das Teilwort aaa enthalten.</solution>
Response: false

Input: <student>Wörter über dem Alphabet {a, b}, die drei as hintereinander enthalten.</student><solution>Wörter über dem Alphabet {a, b}, die das Teilwort aaa enthalten.</solution>
Response: true

Input: <student>beliebig viele as und bs dann drei as dann wieder beliebig viele as und bs</student><solution>Wörter über dem Alphabet {a, b}, die das Teilwort aaa enthalten.</solution>
Response: true

Input: <student>Wörter mit mindestens so vielen bs als as</student><solution>Wörter über dem Alphabet {a, b}, die zuerst beliebig viele as enthalten und danach nur noch mindestens so viele bs enthalten.</solution>
Response: false

Input: <student>Wörter über dem Alphabet {a, b}, die mit a enden oder leer sind.</student><solution>Wörter über dem Alphabet {a, b}, die nicht auf b enden.</solution>
Response: true

Input: <student>ab i</student><solution>Wörter über dem Alphabet {a, b}, die nicht auf b enden.</solution>
Response: false
\end{lstlisting}

\subsection{M2a: Conversion to RE}

\begin{lstlisting}
# Instructions

* You are given a German description of a formal language. Your task is to give a regular expression expressing the same language.
* Only use the following constructs:
  * words
  * concatenation 
  * parentheses
  * "ε" (empty word)
  * "+" (union)
  * "*" (Kleene star)
  * "∅" (empty set)
* The alphabet is implicitly fixed to %ALPH%, unless specified otherwise.
* If possible, try to maintain the structure of the input rather than making simplifications. For example, "das leere Wort oder beliebig viele as" should become "ε+a*", even though "a*" is a more clear and concise expression.
* Think about your solution. Then return the bare regular expression as your output without any additional comments.
* If the input is not a language description or describes a non-regular language, return the exact string "ERR".

# Examples

Input: Wörter über dem Alphabet {a, b}, die genau drei as enthalten.
Response: a*ba*ba*ba*

Input: Wörter die mit dem selben Symbol aufhören und anfangen
Response: a(a+b)*a+b(a+b)*b

Input: Mindestens einfache Wiederholung von ab
Response: ab(ab)*

Input: ab i
Response: ERR
\end{lstlisting}

\subsection{M2b: Conversion to CFG}

\begin{lstlisting}
# Instructions

* You are given a German description of a formal language. Your task is to give a context-free grammar expressing the same language.
* Use the syntax found in the examples.
* The starting symbol should be named S and be the first rule.
* The alphabet is implicitly fixed to %ALPH%, unless specified otherwise.
* If possible, try to maintain the structure of the input rather than making simplifications. For example, "ein a oder beliebig viele as" could become "S -> a | aS | ε" rather than "S -> aS | ε", even though the latter is more concise.
* Think about your solution. Then return the bare grammar as your output without any additional comments.
* If the input is not a language description or describes a language that is not context-free, return the exact string "ERR".

# Examples

Input: Wörter über dem Alphabet {a, b}, die mehr as als bs enthalten.
Response:
S -> aSb | bSa | SS | aS | a

Input: Wörter die mit dem selben Symbol aufhören und anfangen
Response:
S -> aAa | bAb
A -> aA | bA | ε

Input: Die leere Sprache
Response:
S -> S

Input: Wörter die mit beliebig vielen as beginnen, dann beliebig viele bs haben, und dann auf beliebig viele cs enden
Response:
S -> aS | A
A -> bA | B
B -> cB | ε

Input: ab i
Response: ERR
\end{lstlisting}

\subsection{M3: Conversion to \Nile}

\begin{lstlisting}
# Instructions

* Your task is to convert German descriptions of formal languages into a formal representation called NILE.
* NILE expressions are built from the operators from propositional logic (&, |, !, ...), concatenation (;), number predicates (=, >=, ...), and many additional constructs (REP, HAS, BEGIN, ...). Only use constructions present in the examples.
* The alphabet is implicitly fixed to %ALPH%, unless specified otherwise.
* If possible, maintain the structure of the input rather than making simplifications. Examples:
  * "beginnt mit mindestens einem a" should become "BEGIN(>=1,a)" (instead of BEGIN(a)).
  * "enthält genau ein a und beliebig viele bs" should become "HAS(=1,a)&HAS(>=0,b)" (instead of HAS(=1,a)).
* Many natural-language phrasings can't easily be encoded into NILE as in the examples above. When in doubt, prioritize accurate semantics over maintaining syntactic similarity.
* Think about your solution. Then return the bare NILE expression as your output without any additional comments.
* If the input is not a language description, return the exact string "ERR".

## Basics / Propositional Logic

* NL: Die Sprache, die nur aus dem Wort "ba" besteht.
* NILE: ba

* NL: Die Sprache, die nur aus dem leeren Wort besteht.
* NILE: EPS

* NL: Die Sprache, die alle Wörter enthält.
* NILE: TRUE

* NL: Die leere Sprache.
* NILE: FALSE

* NL: Die Sprache der Wörter, die nicht das Teilwort ab enthalten.
* NILE: !HAS(ab)

* NL: Die Sprache der Wörter, die sowohl a als auch b enthalten.
* NILE: HAS(a)&HAS(b)

* NL: Die Sprache der Wörter, die a enthalten oder b enthalten.
* NILE: HAS(a)|HAS(b)

* NL: Die Sprache der Wörter, die b enthalten, wenn sie a enthalten.
* NILE: HAS(a)->HAS(b)

* NL: Die Sprache der Wörter, die genau dann a enthalten, wenn sie b enthalten.
* NILE: HAS(a)<->HAS(b)

* NL: Die Sprache der Wörter, die sowohl entweder ab enthalten oder ba enthalten, als auch bb enthalten.
* NILE: (HAS(ab)|HAS(ab))&HAS(b)
* Comment: Parentheses are necessary whenever combining different binary operators.

* NL: Die Sprache, die nur die Wörter a und b enthält.
* NILE: a|b
* Comment: Note that "&" here does not translate the meaning of "und". Compare the next example.

* NL: Die Sprache der Wörter, die sowohl a als auch b sind.
* NILE: a&b

## Concatenation

* NL: Die Sprache der Wörter, die mit a beginnen und danach nur noch genau b enthalten.
* NILE: a;b
* Comment: Although the description contains "beginnen" and "enthalten", the primitives "BEGIN" and "HAS" can't easily be used here. The underlying mechanism here is simple concatenation, so return that.

* NL: Die Sprache der Wörter, die mit a anfangen und danach nur noch beliebig viele bs enthalten.
* NILE: a;REP(b)

* NL: Die Sprache der Wörter, die a enthalten und danach b enthalten.
* NILE: HAS(a);HAS(b)

## Repetition

* NL: Die Sprache der Wörter, die aus einer beliebigen Folge aus as bestehen.
* NILE: REP(a)

* NL: Die Sprache der Wörter, die beliebige Wiederholungen von a und b sind.
* NILE: REP(a|b)
* Comment: The next two examples are synonymous with this.

* NL: Die Sprache aller Wörter über dem Alphabet {a, b}.
* NILE: ALPH(a,b):TRUE

* NL: Die Sprache der Wörter, die nur aus a und b bestehen.
* NILE: CONS(a,b)
* Comment: Same semantics as "CONS(>=0,a,>=0,b)".

* NL: Die Sprache der Wörter über dem Alphabet {a, b}, die entweder a enthalten oder zwei Zeichen lang sind.
* NILE: ALPH(a,b):HAS(a)|LEN=2

## Constructs based on concatenation and repetition

* NL: Die Sprache der Wörter, die das Teilwort aba enthalten.
* NILE: HAS(aba)

* NL: Die Sprache der Wörter, die mit a beginnen.
* NILE: BEGIN(a)

* NL: Die Sprache der Wörter, die auf a enden.
* NILE: END(a)

* NL: Die Sprache der Wörter, die mit a beginnen und enden.
* NILE: BEGIN(a)&END(a)

* NL: Die Sprache der Wörter, die mit ba oder aab beginnen.
* NILE: BEGIN(ba|aab)

* NL: Die Sprache der Wörter, die mit einem Teilwort beginnen, welches mit a beginnt, worauf nur noch ein Teilwort folgt, welches b enthält.
* NILE: BEGIN(a);HAS(b)
* Comment: Illustrates a possible wrong way to translate "Wörter, die mit a beginnen und danach nur noch b enthalten".

* NL: Die Sprache der Wörter, in denen a und b an abwechselnden Positionen stehen.
* NILE: ALTERNATE(a,b)

## Predicates

* NL: Die Sprache der Wörter, die genau zwei Zeichen enthalten.
* NILE: LEN=2

* NL: Die Sprache der Wörter, die nicht genau zwei Zeichen enthalten.
* NILE: LEN!=2
* Comment: "!=" isn't a new operator; it's actually just negation preceding "=". "!>" or "!EVEN" are also valid predicates.

* NL: Die Sprache der Wörter, die mindestens zwei Zeichen enthalten.
* NILE: LEN>=2

* NL: Die Sprache der Wörter, die mehr als zwei Zeichen enthalten.
* NILE: LEN>2

* NL: Die Sprache der Wörter, die höchstens zwei Zeichen enthalten.
* NILE: LEN<=2

* NL: Die Sprache der Wörter, die weniger als zwei Zeichen enthalten.
* NILE: LEN<2

* NL: Die Sprache aller Wörter mit gerader Länge.
* NILE: LEN EVEN

* NL: Die Sprache aller Wörter mit ungerader Länge.
* NILE: LEN ODD

* NL: Die Sprache aller Wörter, deren Länge durch zwei teilbar ist.
* NILE: LEN=0 MOD 2

* NL: Die Sprache aller Wörter, deren Länge nicht größer als zwei ist.
* NILE: LEN!>2

* NL: Die Sprache aller Wörter mit einer Länge von 2 bis einschließlich 5.
* NILE: LEN IS(>=2&<=5)
* Comment: The "IS" is necessary with free-standing conjunctions/disjunctions of predicates for syntax reasons, as "LEN>=2&<=5" would group into "(LEN>=2)&(<=5)", and "LEN(>=2&<=5)" would try to pass the predicate as an argument.

## Repetition with predicates

* NL: Die Sprache der Wörter, die aus einer mindestens einfachen Wiederholung des Wortes ab bestehen.
* NILE: REP(>=1,ab)

* NL: Die Sprache der Wörter, die aus einer Folge aus mindestens null as bestehen.
* NILE: REP(>=0,a)

* NL: Die Sprache der Wörter, die aus dem Teilwort ab bestehen, welches mindestens einmal und gerade oft wiederholt wird.
* NILE: REP(>=1&EVEN,ab)
* Comment: Here, since the predicate isn't free-standing, the "IS" isn't necessary.

* NL: Die Sprache der Wörter, die aus exakt einem a und beliebig vielen bs bestehen.
* NILE: CONS(=1,a,>=0,b)

* NL: Die Sprache der Wörter, die eine ungeraden Anzahl von as enthalten.
* NILE: HAS(ODD,a)
* Comment: HAS/COUNT count in terms of character positions, so "aaa" is considered to contain the subword "aa" twice, even though it's only possible to "pack" one.

* NL: Die Sprache der Wörter, die beliebig viele as enthalten.
* NILE: HAS(>=0,a)

* NL: Die Sprache der Wörter, die mindestens ein a enthalten.
* NILE: HAS(>=1,a)

* NL: Die Sprache der Wörter, die mit genau 2 as beginnen.
* NILE: BEGIN(=2,a)

* NL: Die Sprache der Wörter, die mit mindestens einem a beginnen.
* NILE: BEGIN(>=1,a)

* NL: Die Sprache der Wörter, die auf mindestens ein a enden.
* NILE: END(>=1,a)

* NL: Die Sprache der Wörter, wo die Anzahl der as und bs zusammen höchstens 2 ist.
* NILE: COUNT(a|b)<=2

* NL: Die Sprache der Wörter, wo die Anzahl der as gleich der Wortlänge ist.
* NILE: COUNT(a)=LEN

* NL: Die Sprache der Wörter mit höchstens doppelt so vielen bs wie as.
* NILE: COUNT(b)<=2*COUNT(a)

## Quantification

* NL: Die Sprache der Wörter, die mit a beginnen und danach nur noch i bs enthalten, wobei i eine beliebige natürliche Zahl ist.
* NILE: EX(i):a;REP(=i,b)

* NL: Die Sprache der Wörter, die mit beliebig vielen as beginnen und danach nur noch gleich viele bs enthalten.
* NILE: EX(i):REP(=i,a);REP(>=i,b)

* NL: Die Sprache der Wörter, die mit mindestens einem a beginnen und mit gleich vielen bs enden.
* NILE: EX(i,>=1):BEGIN(=i,a)&END(=i,a)

## Indexing

* NL: Die Sprache der Wörter, deren drittes Zeichen ein a ist.
* NILE: AT(3,a)
* Comment: Indices are 1-based in keeping with natural language.

* NL: Die Sprache der Wörter, wo an allen geraden Positionen ein a steht.
* NILE: FORALL(i,POS):i EVEN->AT(i,a)
* Comment: The "POS" predicate is true for all word positions (i.e. 1 to LEN). It is necessary here because the claim is evidently not true for out-of-bounds indices like 0.

* NL: Die Sprache der Wörter gerader Länge, in deren erster Hälfte keine as vorkommen.
* NILE: EX(i):LEN=2*i&RANGE(1,i,HAS(=0,a))
* Comment: RANGE is inclusive.

* NL: Die Sprache der Wörter, die aus gleich vielen ds und bs bestehen und wo in jedem Präfix mindestens so viele ds wie bs vorkommen.
* NILE: CONS(d,b)&COUNT(d)=COUNT(b)&(FORALL(i,POS):RANGE(1,i,COUNT(d)>=COUNT(b)))

## Word operations

* NL: Die Sprache der Palindrome.
* NILE: PALINDROME

* NL: Die Sprache der Wörter, die aus einer mindestens dreifachen Wiederholung eines beliebigen Wortes bestehen.
* NILE: EX(w):REP(>=3,w)
* Comment: Whether a variable should refer to a word or a number is determined from context.

* NL: Die Sprache der Wörter, die mit einem beliebigen Teilwort beginnen, worauf danach nur noch das selbe Teilwort rückwärts gelesen folgt.
* NILE: EX(w):w;REVERSE(w)

* NL: Die Sprache der Wörter, die rückwärts gelesen mit ab beginnen.
* NILE: REVERSE(BEGIN(ab))
\end{lstlisting}
\fi

\end{document}